\documentclass{article}
\usepackage{times}
\usepackage{soul}
\usepackage{url}
\usepackage[hidelinks]{hyperref}
\usepackage[utf8]{inputenc}
\usepackage[small]{caption}
\usepackage{graphicx}
\usepackage{amsmath}
\usepackage{amsthm}
\usepackage{booktabs}
\usepackage{algorithm}
\usepackage{algorithmic}
\usepackage[]{natbib}
\usepackage{xcolor}

\DeclareMathOperator*{\argmax}{arg\,max}
\definecolor{myred}{rgb}{0.6,0,0}
\definecolor{myblue}{rgb}{0,0,0.6}

\newtheorem{definition}{Definition}
\newtheorem{example}{Example}
\newtheorem{theorem}{Theorem}

\newtheorem{lemma}{Lemma}
\newtheorem*{lemma*}{Lemma}
\newtheorem*{theorem*}{Theorem}
\newtheorem{observation}{Observation}

\usepackage{tikz}
\usetikzlibrary{calc}
\usetikzlibrary{patterns}

\date{}
\title{On the Complexity of Majority Illusion in Social Networks}

\author{
	Umberto Grandi \\ University of Toulouse \and Grzegorz Lisowski \\ University of Warwick \and M.S. Ramanujan \\ University of Warwick \and Paolo Turrini \\ University of Warwick
}

\usepackage{subfigure}
\usepackage{tikz}
\usetikzlibrary{backgrounds}
\usetikzlibrary{calc,shapes,snakes,matrix,arrows} 

\usetikzlibrary{calc}
\usetikzlibrary{fit}
\tikzset{
	between/.style args={#1 and #2}{
		at = ($(#1)!0.5!(#2)$)
	}
}

\usepackage{relsize}

\tikzset{fontscale/.style = {font=\relsize{#1}}
}

\usepackage{verbatim}
\usepackage{amsmath,amsfonts}

\usepackage{fancyhdr}
\begin{document}
	\maketitle
	
	\begin{abstract}
		Majority illusion occurs in a social network when the majority of the network nodes belong to a certain type but each node's neighbours mostly belong to a different type, therefore creating the wrong perception, i.e., the illusion, that the majority type is different from the actual one. From a system engineering point of view, we want to devise algorithms to detect and, crucially, correct this undesirable phenomenon. In this paper we initiate the computational study of majority illusion in social networks, providing complexity results for its occurrence and avoidance. Namely, we show that identifying whether a network can be labelled 
		such that majority illusion is present, as well as the problem of removing an illusion by adding or deleting edges of the network, are NP-complete problems.
	\end{abstract}
	
	\section{Introduction}
	
	Social networks shape the way people think. Individuals' private opinions can change as a result of social influence and a well-placed minority view can become what most people come to believe  \citep{Stewart2019}. 
	There is also a natural tendency for people to connect to individuals similar to them, the so-called homophily (see, e.g. \cite{mcpherson2001birds}), which adds to the potential for a social network to create information bubbles and is amplified even further in modern social media networks (\cite{lee2019homophily}). 
	The current vaccination debate has brought to the fore the dramatic effects that misperception can have in people's lives  \citep{covid-nature} and made it clear how important it is to design social networks where participants receive the most unbiased information possible.
	
	When individuals use their social network as a source of information, it can happen that minority groups are more ``visible" as a result of being better placed, which makes them overrepresented in many friendships' groups. Sometimes these minorities can be so well placed that many or even most individuals ``see" them as majorities - a phenomenon called {\em majority illusion}. Majority illusion was originally introduced by \citet{lerman2016majority} who studied the existence of social networks in which most agents belong to a certain binary type, but most of their peers belong to a different one. Thus, they acquire the wrong perception, i.e., the illusion, that the majority type is different from the actual one. 
	Figure \ref{fig:deadlock} gives an example of this.
	
	
	
	
	\begin{figure}[H]	
		\centering
		\scalebox{0.8}{
			\begin{tikzpicture}
				[->,shorten >=1pt,auto,node distance=1.2cm,
				semithick]
				\node[shape=circle,draw=myred, pattern=crosshatch, pattern color=myred] (A)  {};
				\node[shape=circle,draw=myred,  right of=A, pattern=crosshatch, pattern color=myred] (B)  {};
				\node[shape=circle,draw=myred, below of=A, pattern=crosshatch, pattern color=myred] (C)  {};
				\node[shape=circle,draw=myred,  below of=B, pattern=crosshatch, pattern color=myred] (D)  {};
				
				\draw [thick,-] (B) to  (A) ;
				\draw [thick,-] (C) to  (A) ;
				\draw [thick,-] (D) to  (A) ;
				\draw [thick,-] (C) to  (B) ;
				\draw [thick,-] (C) to  (D) ;
				\draw [thick,-] (B) to  (D) ;

				\node[shape=circle,draw=myblue,  left of=C, fill=myblue] (E)  {};
				\node[shape=circle,draw=myblue, above of=A, fill=myblue] (E')  {};
				\node[shape=circle,draw=myblue, above of=B, fill=myblue] (F)  {};
				\node[shape=circle,draw=myblue, above of=B, right of=B, fill=myblue] (G)  {};
				
				\node[shape=circle,draw=black, right of=D, fill=myblue] (H)  {};
				
				\draw [thick,-] (E') to  (F) ;
				\draw [thick,-] (F) to  (A) ;
				\draw [thick,-] (E') to  (B) ;
				\draw [thick,-] (D) to  (H) ;

				\draw [thick,-] (E) to  (C) ;
				\draw [thick,-] (E') to  (A) ;
				
				\draw [thick,-] (F) to  (B) ;
				\draw [thick,-] (G) to  (B) ;
				
				
		\end{tikzpicture}}

		\caption{An instance of majority illusion. The well-placed red minority is seen as a majority by everyone}\label{fig:deadlock}

	\end{figure}
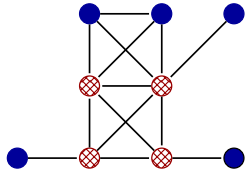

	Majority illusion has important consequences when paired with opinion formation. If for example individuals are influenced by the majority of their friends to change their mind, i.e., they abide to the well-known threshold model (\cite{granovetter1978threshold}), then majority illusion means that the overrepresented minorities become stable majorities.
	As such it is important to predict the occurrence of majority illusions in a network and, crucially, how a given network can be transformed so that this undesirable phenomenon is eliminated.
	
	Some analysis of majority illusion is already present in the literature.  \citet{lerman2016majority}, for example, studied network features that correlate with having many individuals under illusion. In particular, the study demonstrated how disassortative networks, i.e. those in which highly connected agents tend to link with lowly connected ones, increase the chances of majority illusion. 
	However, the computational questions of checking whether a network admits majority illusion and, crucially, how this can be corrected, are still unanswered. 
	
	Network transformation has shown important applications in the context of election manipulation (see, e.g., \cite{castiglioni2021election}), influence maximisation  \citep{zhou21maximizing}, anonymisation (see, e.g., \cite{kapron2011social}) and of $k$-core maximization (see, e.g.,  \cite{chitnis}, \cite{zhou2019k}).  Applying optimal network transformation techniques for illusion elimination is therefore a natural and important challenge. 

	

	\paragraph{Our contribution.}
	In this paper we initiate the algorithmic analysis of majority illusion in social networks, focusing on two computational questions.
	First, we are interested in which networks allow for the possibility of illusion, i.e., whether there is a labelling of the nodes such that a specified fraction of agents is under illusion. We show that such problem is NP-complete for every fraction strictly greater than $\frac{1}{2}$ by a non-trivial reduction from the NP-complete problem 3-SAT. 
	Further, we focus on the problem of eliminating illusion from a network, by modifying the agents' connectivity with a constraint on the number of edges which can be added or eliminated. We show that checking if it is possible to alter the structure of the network to ensure that at most a given fraction 
	of agents is under illusion is NP-complete, reducing from the NP-complete problem 2P2N-SAT.

	\paragraph{Other related work.} Our results are also connected to a number of research lines in various AI-related areas.

	{\em Opinion Manipulation.} Our work is directly related to computational models of social influence, notably the work of \citet{AulettaEtAlAIJ2020}, where networks and initial distribution of opinions  are identified such that an opinion can become a consensus opinion following local majority updates. In this context, it is important to observe that when all nodes are under majority illusion, a synchronous majoritarian update causes an initial minority to evolve into a consensus in just one step. Other notable models include  \cite{DoucetteEtAlJAAMAS2019} who studied the propagation of possibly incorrect opinions with an objective truth value in a social network, and the stream of papers studying the computational aspects of exploiting (majoritarian) social influence via opinion transformation \citep{BredereckElkind2017,AulettaEtAlAIJ2020,AulettaEtAlTCS2021,CastiglioniAAAI2020}.
	
	{\em Network Manipulation.} An important research line has looked at how to transform a social network structure with applications in the voting domain. \cite{WilderVorobeychik2018}, e.g., studied how an external manipulator having a limited budget can select a set of agents to directly influence, to obtain a desired outcome of elections. In a similar setting, \cite{FaliszewskiEtAl2018} studied `` bribes" of voters' clusters. 
	
	{\em Social Choice on Social Networks.}
	Our research aligns with the work in computational social choice, in particular strategic voting \citep{MeirStrategicVoting} and iterative voting (e.g., \cite{MeirAIJ2017,ReijngoudEndrissAAMAS2012}) where decision-making happens sequentially. Of relevance are also the recently found connections between iterative voting and social networks (\cite{Wilczynski2019PollConfidentVI}, \cite{Baumeister2020ManipulationOO}).

	There are also various other accounts of paradoxical effects in social networks which are related to our work, such as  the \emph{friendship paradox}, according to which, on average, individuals are less well-connected than their friends (see, e.g.\cite{hodas2013friendship}, \cite{alipourfard2020friendship}). Exploiting a similar paradox, \cite{santos2021biased} recently showed how false consensus leads to the lack of participation in team efforts.
	
	\paragraph{Paper structure.}
	Section~\ref{sec:preliminaries} provides the basic setup and definitions. Section~\ref{sec:verifying} focuses on checking whether illusion can occur in a network while Section~\ref{sec:eliminating} studies illusion elimination. Section~\ref{sec:conclusions} concludes the paper presenting various potential future directions. Some proofs are omitted and can be found in the appendix.

	
	\section{Preliminaries}\label{sec:preliminaries}
	
	Our model features a set $N$ of agents, connected in a graph $(N,E)$, with $E\subseteq N^2$. Throughout the paper we will consider \emph{undirected graphs}, requiring $E$ to be symmetric. Furthermore, we assume that $E$ is \emph{irreflexive}, i.e. that $E$ does not include self-loops. We call such a graph a \emph{social network}. For $i\in N$ we denote as $E(i)=\{j \in N : E(i,j)\}$ the set of agents that $i$ is following. 
	Furthermore, a network  $( N, E)$ is an \emph{extension} of $( N, E')$ if $E' \subseteq E$. Similarly, if $E \subseteq E'$, we say that $(N,E)$ is a \emph{subnetwork} of $(N,E')$.
	
	\paragraph{Labellings.}
	
	We will work with social networks where each of the agents has an opinion, which we model as a labelling (or a colouring) over two possible alternatives. 	So, we consider \emph{labelled social networks}, in which every node is assigned its alternative (colour). Throughout the paper we assume a binary set of colours $C=\{b,r\}$ (\emph{blue} and \emph{red}). 
	
	\begin{definition}[Labelled Social Network]
		A \emph{labelled social network} is a tuple $(N,E,f)$, where $(N,E)$ is a social network and $f: N \rightarrow C$ is a \emph{labelling} which assigns an alternative to each agent.
	\end{definition}
	
	Further, given a labelling $f$ of a social network $(N,E)$, we denote the set of red nodes $\{i \in N :  f(i)=r \} $ as $R_f$ and the set of blue nodes $\{i \in N :  f(i)=b \}$ as $B_f$. Moreover, for a set $S \subseteq N$,  $R^S_f$ is the set of red nodes in $S$, while $B^S_f$ is the set of blue nodes in $S$. We omit $f$ if clear from the context.

	We will further distinguish between the majority option in the entire social network and the majority option from an agent's perspective, while only considering \emph{strict} majority. It is worth noting that under such a definition, a majority winner does not exist if the number of nodes labelled blue is the same as of those labelled red. So, given a labelled social network \textit{SN}=$(N,E,f)$, we denote the colour adopted by the strict majority in  \textit{SN} as the \emph{majority winner} ($W_{\textit{SN}}$). Formally, a colour $c$ is a majority winner in \textit{SN} if and only if $| \{n \in N : f(n)=c   \} | > | \{ n' \in N : f(n') \neq c  \}   |$.
	Similarly, for an agent $i$, $W_{\textit{SN}}^i$ is the majority option in $i$'s (open) neighbourhood. Formally, a colour $c$ is a majority winner in $i$'s neighbourhood if and only if $| \{n \in E(i) : f(n)=c   \} | > | \{ n' \in E(i) : f(n') \neq c  \}   |$. Henceforth, where relevant, we will assume without loss of generality that blue is the majority winner in a network.
	
	We are now ready to define the concept of \emph{majority illusion}, that occurs when a certain number of agents has a wrong perception of which colour is the majority winner in the network.
	We say that an agent $i \in N$ is \emph{under illusion} if $W_{\textit{SN}}$ and $W^i_{\textit{SN}}$ exist, while $W_{\textit{SN}}^i \neq W_{\textit{SN}}$. 

	\begin{definition}[$q$-majority illusion]
		Let $q \in \mathbb{Q} \cap [0,1]$. Then, a \emph{$q$-majority illusion} is a labelled social network $\textit{SN}=(N,E,f)$ such that at least $q \cdot |N| $ agents are under illusion. 
	\end{definition}
	
	For a given social network $(N,E)$, fraction $q$ and a function $f: N \rightarrow C$, we say that $f$ \emph{induces} a $q$-majority illusion, if $(N,E,f)$ is a $q$-majority illusion. 
	When not confusing, we will sometimes only say that $f$ induces illusion. If there is a labelling of a network \textit{SN} which induces $q$-majority illusion, then we say that \textit{SN} \emph{admits} a $q$-majority illusion. Also, for a network $(N,E)$ and $n,n' \in N$ such that $E(n)={n'}$, we say that $n$ is a \emph{dependant} of $n'$. Let us further observe, that if a labelling $f$ induces 1-majority illusion for a network $(N,E)$ and $n$ is a dependant of $n'$, then $f(n')=r$. Finally, for a labelled network $(N,E,f)$ and $i\in N$ we define the \emph{margin of victory} for $i$ as $|B^{E(i)}_f| - |R^{E(i)}_f|$.
	

	\section{Verifying Illusion}\label{sec:verifying}

	We are interested in finding the complexity of checking, for a specific $q$, if a given network admits a $q$-majority illusion. 
	
	\begin{quote}
		\noindent \textsc{$q$-majority illusion}:\\
		\hspace*{-1em} \indent\textit{Input:} Social network $\textit{SN}=( N, E)$.\\
		\hspace*{-1em}\textit{Question:} Is there a colouring $f: N \rightarrow C$ such $f$ induces a $q$-majority illusion?
	\end{quote}

	We now prove that \textsc{$q$-majority illusion} is NP-complete for every rational  $q \in(\frac{1}{2},1]$, by providing a reduction from the NP-complete problem \textsc{3-SAT} for every such $q$. In \textsc{3-SAT} we check the satisfiability of a CNF formula in which all clauses have exactly three literals (see, e.g. \cite{papadimitriou2003computational}). We say that such a formula is in 3-CNF.
	We describe the constructions and sketch the main lines of the proof, which can be found in complete form in the appendix.
	
	Let $\varphi$ be a formula in 3-CNF. We commence with constructing a social network which we call the \emph{encoding} of $\varphi$, or $E_{\varphi}=(N,E,f)$. We will further show that it admits 1-majority illusion if and only if $\varphi$ is satisfiable, entailing the NP-hardness of \textsc{1-majority illusion}. Finally, for each  $q \in(\frac{1}{2},1]$ we construct $E^q_{\varphi}$ appending a non-trivial network construction to $E_{\varphi}$. We then conclude the proof showing that $E^q_{\varphi}$ admits a $q$-majority illusion iff $\varphi$ is satisfiable. 
	
	
	\paragraph{Variable, clause, and balance gadgets.} 
	{For a formula $\varphi$ in 3-CNF, we denote the set of variables in $\varphi$ as $P_{\varphi}=\{p_1, \dots, p_m\}$, and the set of clauses in $\varphi$ 
		as $C_{\varphi}=\{C_1, \dots, C_n\}$.}
	The first step is to encode propositional variables. 
	For a variable $p_i$, we define a subnetwork called \emph{variable gadget} as depicted in Figure \ref{VariableGadget1}. {We refer to the nodes in the bottom pair of the gadget as \emph{literal nodes}. Also, we call the left literal node $p_i$, and the right $\neg p_i$.} 
	
	\begin{lemma}
		A labelling of a variable gadget (considered as a separate network) induces a 1-majority illusion only if exactly one of the nodes in the bottom pair is labelled $r$. 
	\end{lemma}

	We say that a labelling of a variable gadget is of type A if it induces a 1-majority illusion and $f(p_i)=r$. Symmetrically, we say that a labelling is of type B if it induces illusion and $f(\neg p_i)=r$. It is worth to observe that labellings of type A and of type B are unique. 
	
	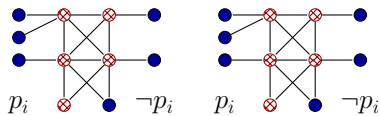
\begin{figure}[H]
		\centering
		
		\scalebox{0.5}{
			
			\begin{tikzpicture}
				[->,shorten >=1pt,auto,node distance=1.2cm,
				semithick]
				\node[shape=circle,draw=myred, pattern=crosshatch, pattern color=myred] (A) {};
				\node[shape=circle,draw=myred, pattern=crosshatch, pattern color=myred] (B) [ right of= A] {};
				\node[shape=circle,draw=myred, pattern=crosshatch, pattern color=myred] (C) [below of = A] { };
				\node[shape=circle,draw=myred, pattern=crosshatch, pattern color=myred] (D) [below of = B] { };
				
				\node[shape=circle,draw=myred, pattern=crosshatch, pattern color=myred] (E) [below of = C] { };
				\node (E') [left of = E] [fontscale=4] { $\displaystyle p_i$ };
				\node[shape=circle,draw=black, fill=myblue] (F) [below of = D] { };
				\node[shape=circle,draw=white] (F') [right of = F] [fontscale=4] {$\neg p_i$ };
				
				\node[shape=circle,draw=black, fill=myblue] (G) [left of = A]{};
				\node[shape=circle,draw=black, fill=myblue] (H) [right of = B]{};
				\node[shape=circle,draw=black, fill=myblue] (I) [ below of=H]{};
				\node[shape=circle,draw=black, fill=myblue] (K) [ below of=G]{};
				\node[shape=circle,draw=black, fill=myblue] (TEST) at ($(G)!0.5!(K)$) {};
				
				\draw [thick,-]  (TEST) to (A) ;
				
				\draw [thick,-]  (A) to(G) ;
				\draw [thick,-]  (K) to(C) ;
				\draw [thick,-]  (B) to(H) ;
				\draw [thick,-]  (D) to(I) ;
				
				\draw [thick,-]  (A) to(B) ;
				\draw [thick,-] (B) to  (C) ;
				\draw [thick,-] (A) to  (D) ;
				\draw [thick,-] (A) to  (C) ;
				\draw [thick,-] (C) to  (D) ;
				\draw [thick,-] (B) to  (D) ;

				\draw [thick,-] (C) to  (E) ;
				\draw [thick,-] (D) to  (F) ;
				\draw [thick,-] (C) to  (F) ;
				\draw [thick,-] (D) to  (E) ;
		\end{tikzpicture}}
		~
		\scalebox{0.5}{
			
			\begin{tikzpicture}
				[->,shorten >=1pt,auto,node distance=1.2cm,
				semithick]
				\node[shape=circle,draw=myred, pattern=crosshatch, pattern color=myred] (A) {};
				\node[shape=circle,draw=myred, pattern=crosshatch, pattern color=myred] (B) [ right of= A] {};
				\node[shape=circle,draw=myred, pattern=crosshatch, pattern color=myred] (C) [below of = A] { };
				\node[shape=circle,draw=myred, pattern=crosshatch, pattern color=myred] (D) [below of = B] { };
				
				\node[shape=circle,draw=myred, pattern=crosshatch, pattern color=myred] (E) [below of = C] { };
				\node[shape=circle,draw=white] (E') [left of = E] [fontscale=4] {$p_i$ };
				\node[shape=circle,draw=black, fill=myblue] (F) [below of = D] { };
				\node[shape=circle,draw=white] (F') [right of = F] [fontscale=4] {$\neg p_i$ };
				
				\node[shape=circle,draw=black, fill=myblue] (G) [left of = A]{};
				\node[shape=circle,draw=black, fill=myblue] (H) [right of = B]{};
				\node[shape=circle,draw=black, fill=myblue] (I) [ below of=H]{};
				\node[shape=circle,draw=black, fill=myblue] (K) [ below of=G]{};
				\node[shape=circle,draw=black, fill=myblue] (TEST) at ($(G)!0.5!(K)$) {};

				\draw [thick,-]  (A) to(G) ;
				\draw [thick,-]  (K) to(C) ;
				\draw [thick,-]  (B) to(H) ;
				\draw [thick,-]  (D) to(I) ;
				
				\draw [thick,-]  (TEST) to(A) ;
				
				\draw [thick,-]  (A) to(B) ;
				\draw [thick,-] (B) to  (C) ;
				\draw [thick,-] (A) to  (D) ;
				\draw [thick,-] (A) to  (C) ;
				\draw [thick,-] (C) to  (D) ;
				\draw [thick,-] (B) to  (D) ;

				\draw [thick,-] (C) to  (E) ;
				\draw [thick,-] (D) to  (F) ;
				\draw [thick,-] (C) to  (F) ;
				\draw [thick,-] (D) to  (E) ;
			\end{tikzpicture}
		}

		\caption{Variable gadget of type A in the left network, and of type B in the right network. The gadgets above correspond to $p_i$, and we refer to the left node in the bottom pair as $p_i$, and to the right as $\neg p_i$.}\label{VariableGadget1}
	\end{figure}
	
	As a second step, we define \emph{clause gadgets}, associated to each clause  $C_j \in C_{\varphi}$, as depicted in Figure~\ref{ClaueGadget1}. {The top three nodes outside of the dashed rectangle are literal nodes and do not belong to the gadget.}
	Then, a clause gadget consists of sixteen nodes 
	, including a 5-clique. 
	In this structure, three members of the clique are adjacent to two dependants each and to  
	one additional node, which we call a \emph{co-dependant}. 
	The three co-dependants form a clique in this gadget. 
	{Members of the five-clique which are adjacent to a co-dependant, are also adjacent to particular literals nodes. For every literal $L$ in the $C_j$, $L$ is adjacent to exactly one of the mentioned members of the clique and at most one literal node is adjacent to each member of a clause gadget.} Connections between literal nodes and a clause gadget are shown in the Figure {\ref{ClaueGadget1}}.
	The remaining two nodes in the 5-clique have one dependant each.

	\begin{figure}[H]
		\centering
		
		\scalebox{0.4}{
			\begin{tikzpicture}
				[->,shorten >=1pt,auto,node distance=1.5cm,
				semithick]
				
				\node[shape=circle,draw=myred, pattern=crosshatch, pattern color=myred] (S)  { };
				\node[shape=circle,draw=myred, pattern=crosshatch, pattern color=myred] (A)[below of=S, left of=S] {};
				\node[shape=circle,draw=myred, pattern=crosshatch, pattern color=myred] (B) [below of=S, right of= S] {};
				\node[shape=circle,draw=myred, pattern=crosshatch, pattern color=myred] (C) [below of = A] { };
				\node[shape=circle,draw=myred, pattern=crosshatch, pattern color=myred] (D) [below of = B] { };
				
				\node[shape=circle,draw=black, fill=myblue] (L2) [above of=S] { };

				\node[shape=circle,draw=black, fill=myblue] (E) [below of = C] { };
				\node[shape=circle,draw=black, fill=myblue] (F) [below of = D] { };
				
				\node[shape=circle,draw=black, fill=myblue] (G) [left of = A]{};
				\node[shape=circle,draw=black, fill=myblue] (G1) [above of = G]{};
				\node[shape=circle,draw=black, fill=myblue] (H) [right of = B]{};
				\node[shape=circle,draw=black, fill=myblue] (H1) [above of = H]{};
				
				\node[shape=circle,draw=myred, pattern=crosshatch, pattern color=myred] (L1) [above of=G1, right of=G1] { };
				\node[] (L1') [left of = L1] [fontscale=4] {$L_1$ };
				\node[shape=circle,draw=black, fill=myblue] (L3) [above of=H1, left of=H1] { };
				\node[shape=circle,draw=white] (L3') at ($(L2)!0.5!(L3)$) [fontscale=4] {$L_3$ };
				\node[shape=circle,draw=white] (L2')at ($(L2)!0.5!(L1)$)  [fontscale=4] {$L_2$ };
				\draw [thick,-]  (L1) to(A) ;
				\draw [thick,-]  (L2) to(S) ;
				\draw [thick,-]  (L3) to(B) ;
				
				\node[shape=circle,draw=black, fill=myblue] (I) [ below of=H]{};
				\node[shape=circle,draw=myred, pattern=crosshatch, pattern color=myred] (J) [ below of=I]{};
				\node[shape=circle,draw=black, fill=myblue] (K) [ below of=G]{};
				\node[shape=circle,draw=black, fill=myblue] (N) [below of = K]{};
				\node[shape=circle,draw=myred, pattern=crosshatch, pattern color=myred] (M) [right of = E]{};
				\draw [thick,-, bend right]  (M) to (J) ;
				\draw [thick,-, bend left]  (M) to (N) ;
				\draw [thick,-, bend left]  (J) to (N) ;

				\draw [thick,-]  (A) to(G) ;
				\draw [thick,-]  (A) to(N) ;
				\draw [thick,-]  (K) to(A) ;
				\draw [thick,-]  (B) to(H) ;
				\draw [thick,-]  (B) to(I) ;
				\draw [thick,-]  (S) to(M) ;
				\draw [thick,-]  (G1) to(S) ;
				\draw [thick,-]  (H1) to(S) ;
				
				\draw [thick,-]  (A) to(B) ;
				\draw [thick,-] (B) to  (C) ;
				\draw [thick,-] (A) to  (D) ;
				\draw [thick,-] (A) to  (C) ;
				\draw [thick,-] (C) to  (D) ;
				\draw [thick,-] (B) to  (D) ;
				\draw [thick,-] (A) to  (S) ;
				\draw [thick,-] (B) to  (S) ;
				\draw [thick,-] (C) to  (S) ;
				\draw [thick,-] (S) to  (D) ;
				\draw [thick,-] (J) to  (B) ;

				\draw [thick,-] (C) to  (E) ;
				\draw [thick,-] (D) to  (F) ;

				\node[draw,  dashed, fit=(G1) (H1) (J) ](FIt1) {};
		\end{tikzpicture}}
		\caption{Clause gadget, corresponding to a clause $(L_1, L_2, L_3)$,  enclosed in the dashed rectangle. The top three nodes are literals}\label{ClaueGadget1}
	\end{figure}
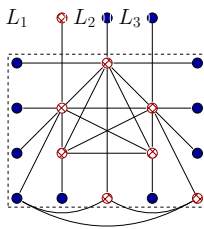

	Observe that in any labelling of this gadget inducing 1-majority illusion all members of the 5-clique need to be labelled $r$, as each of them has a dependant.
	Further, for a labelling of this gadget in which blue is the majority winner to induce 1-majority illusion only two nodes outside of the clique can be labelled $r$. Otherwise, at least eight nodes in the gadget would be labelled $r$ and thus $b$ would not be the unique majority winner. Also, note that at least two co-dependants need to be labelled red in order for all of three of them to be under illusion. So, in a labelling of this gadget which induces 1-majority illusion, exactly 7 nodes are labelled red.
	
	\begin{lemma}
		There exists a labelling of a clause gadget (not as a separate network) which induces 1-majority illusion with blue being a majority winner in this structure if and only if at least one node is adjacent to a literal node labelled $r$. 
	\end{lemma}
	
	
	The final component of the encoding of $\varphi$ is the \emph{balance gadget}. Given a natural number $k \geq 2$, if $k$ is even, it consists of  $\frac{k}{2}$ pairs of nodes. Otherwise, it consists of $\frac{k-1}{2}$ pairs of nodes, and 1 triple of nodes. 

	\paragraph{Encoding of a 3-CNF formula.}
	We are now ready to construct a social network starting from a 3-CNF formula $\varphi$. 
	Firstly, for every $p \in P_{\varphi}$ create a variable gadget as in Figure \ref{VariableGadget1}. Further, for every clause $C_i=\{L_i^1, L_i^2, L_i^3\}$ in $C_{\varphi}$ create a clause gadget as in Figure \ref{ClaueGadget1}, with the literal nodes 
	corresponding to $L_i^1$ adjacent to the top left member of the 5-clique, corresponding to $L_i^2$ to the central top member , and corresponding to $L_i^3$ to the top right member.  As a final step, create a balance gadget with $k= m+2n-1$.
	Observe that as there are $m +2n-1$ nodes in the balance gadget, the total number of nodes in the encoding of $\varphi$ is $12m + 18n -1 $.

	Let us first observe a few facts regarding any labellings of $E_{\varphi}$ for a formula $\varphi$ in 3-CNF which induces a 1-majority illusion. 
	First note that $E_{\varphi}$ contains $m$ variable gadgets, with 11 nodes each. As observed earlier, in every labelling of $E_{\varphi}$ which induces 1-majority illusion, at least 5 nodes have to be labelled $r$ in every variable gadget. Furthermore, $E_{\varphi}$ contains $n$ clause gadgets, with 16 nodes each. Note that in a labelling of $E_{\varphi}$ which induces 1-majority illusion at least 7 nodes need to be labelled $r$ in every clause gadget, as the 5 clique has to be labelled all $r$ due to the presence of dependants, and at least 2  co-dependants need to be labelled red, as otherwise some of the nodes in the bottom 3-clique would not be under illusion. Observe further that in all labellings of the encoding of $\varphi$ that induce a 1-majority illusion, all $m + 2n -1$ members of the balance gadget are labelled red. 
	Hence, due to the presence of the balance gadget, any labelling of $E_{\varphi}$ which induces a 1-majority illusion contains at least $6m+9n-1$ red nodes and at most $6m+9n$ blue nodes, while 
	blue has at most a margin of victory of 1. 
	
	\begin{lemma}
		In a labelling of $E_{\varphi}$ which induces a 1-majority illusion every variable gadget is  of type A or type B.
	\end{lemma}
	
	So, every labelling of $E_{\varphi}$ which induces 1-majority illusion corresponds to a unique valuation over $P_{\varphi}$, where a variable $p_i$ is said to be true if the labelling of the variable gadget corresponding to $p_i$ is of type A, and false if it is of type B. Note also that, as we argued before, a labelling of $E_{\varphi}$ can only induce 1-majority illusion if at least one node in every clause gadget is adjacent to a literal node labelled $r$. 
	Finally, observe that if every variable gadget is of type A or type B, and least one node in every clause gadget is adjacent to a literal node labelled $r$, we can find a labelling of $E_{\varphi}$ which induces illusion, as depicted in Figures \ref{VariableGadget1} and \ref{ClaueGadget1}, with all nodes in the balance gadget labelled red. 
	We are now ready to show that for every formula $\varphi$ in 3-CNF, $E_{\varphi}$ admits 1-majority illusion if and only if $\varphi$ is satisfiable.

	\begin{lemma}\label{lemma:1-hard}
		Let $\varphi$ be a formula in 3-CNF. Then, $\varphi$ is satisfiable if and only if $E_{\varphi}$ admits 1-majority illusion.
	\end{lemma}

	\begin{proof} 
		Let us consider a formula $\varphi$ in 3-CNF with the set of variables $P_{\varphi}=\{p_1, \dots, p_m\}$ and the set of clauses $C_{\varphi}=\{C_1, \dots, C_n\}$. Then, let us construct the encoding $E_{\varphi}$ and show that it admits 1-majority illusion if and only if $\varphi$ is satisfiable.
		Suppose that it is. Then, take a model $M$ of $\varphi$ and construct the following labelling of $E_{\varphi}$. Colour variable gadgets so that for a gadget corresponding to $p_i$, it is of type A if if $p_i$ is true in $M$, and of type B otherwise. Note that, as $M$ is a model of $\varphi$,  by construction of  $E_{\varphi}$ at least one node in every clause gadget is adjacent to a literal node labelled $r$. So, there is a labelling of  $E_{\varphi}$ which induces 1-majority illusion. 
		Further, suppose that $\varphi$ is not satisfiable. Then, assume towards contradiction, that there is a labelling $f$ of  $E_{\varphi}$ which induces 1-majority illusion. Observe that as $f$ induces 1-majority illusion, it corresponds to a unique valuation $V$ over $P_{\varphi}$, where a variable $p_i$ is true in $V$ if it's corresponding gadget is labelled in type A, and false if it is labelled in type B. Furthermore, observe that as $\varphi$ is not satisfiable, there exists a clause $C_j \in C_{\varphi}$ such that for every literal $L$ in $C_j$, $L$ is false in $V$. But this entails that all literal nodes adjacent to the clause gadget corresponding to $C_{\varphi}$ are labelled $b$. But then $f$ does not induce 1-majority illusion, which contradicts the assumptions.
	\end{proof}
	
	We now show some further properties of $E_{\varphi}$.
	Given a 3-CNF formula $\varphi$, let $I_{\varphi} = 6m +9n -1$, where $m$ is the number of variables and $n$ the number of clauses in $\varphi$. Observe that this is the maximum number of nodes which can be labelled red in $E_{\varphi}$ if blue is the strict majority colour in this network.

	\begin{lemma}\label{lemma:RigidEncoding}
		For every 3-CNF formula $\varphi$, $k \leq I_{\varphi}$ and any labelling $f$ of $E_{\varphi}$ such that $R_f = I_{\varphi} - k$, the number of nodes under illusion in $E_{\varphi}$ under $f$ is at most $| N| - k $.
	\end{lemma}

	We also need the following technical lemma.

	\begin{lemma}\label{lemma:h-exists}
		Let $q$ be a rational number in  $(\frac{1}{2}, 1]$, and $k>0$ be a natural number. Then, there exists a natural number $h^*$ such that $\frac{k+ h^*}{k+ 2h^*  }  \geq q$, but $\frac{k+ h^* -1}{k+ 2h^*  } < q $. 
	\end{lemma}

	We refer to such a number as $h^*_{k,q}$. It is not difficult to show that we can compute $h_{k,q}^*$ in polynomial time. This observation is crucial to ensure that the intended reduction is constructable in polynomial time.
	
	We are now ready to prove the main result of this section. To show that \textsc{$q$-majority illusion} is NP-hard for a particular, rational $q$ in  $(\frac{1}{2},1]$, we construct a network $E^q_{\varphi}$ for every formula $\varphi$ in 3-CNF. We start with constructing $E_{\varphi}$ and set of $h_{|E_{\varphi}|, q}^*$ pairs of nodes. Then, it follows from Lemma \ref{lemma:1-hard}, as well as Lemmata \ref{lemma:RigidEncoding} and \ref{lemma:h-exists} that $E^q_{\varphi}$ admits $q$-majority illusion if and only if $\varphi$ is satisfiable. The details of the proof can be found in the appendix. Observe further that \textsc{$q$-majority illusion} is in NP, as one can easily check the number of nodes under illusion in a labelled network. 
	This concludes the proof of the following theorem:

	\begin{theorem}\label{thm:verifyillusion}
		\textsc{$q$-majority illusion} is NP-complete for every rational $q$ in  $(\frac{1}{2},1]$.
	\end{theorem}
	
	
	\section{Eliminating Illusion}\label{sec:eliminating}

	We now turn to the problem of reducing the number of nodes under illusion in a given labelled network, by modifying the connections between them.  
	Namely, we consider the problem of checking if it is possible to ensure that a $q$-majority illusion does not hold in a labelled network be altering only a bounded number of edges.

	\begin{quote}
		\noindent \textsc{$q$-Illusion Elimination}:\\
		\hspace*{-1em} \indent\textit{Input:}  $\textit{SN}=( N, E, f)$ such that $f$ induces $q$-majority illusion in \textit{SN}, $k \in \mathbb{N}$ such that $k \leq |E|$.\\
		\hspace*{-1em}\textit{Question:} Is there a $\textit{SN}'=(N,E',f)$ such that $| \{( e \in N^2: \ e \in E \textit{ iff } e \notin E'  \}   | \leq k$ and  $f$ does not induce $q$-majority illusion in \textit{SN'}?
	\end{quote}
	Subsequently, we consider the problem of eliminating a $q$-majority illusion just by adding edges to the network.
	\begin{quote}
		\noindent \textsc{Addition $q$-Illusion Elimination }:\\
		\hspace*{-1em} \indent\textit{Input:}  $\textit{SN}=( N, E, f)$ s.t. $f$ induces $q$-majority illusion in \textit{SN}, $k \in \mathbb{N}$ such that $k \leq |E|$.\\
		\hspace*{-1em}\textit{Question:} Is there a  $\textit{SN}'=(N,E',f)$ such that \textit{SN} is a subnetwork of \textit{SN'},  $ |E'| - |E| \leq k$ and  $f$ does not induce $q$-majority illusion in \textit{SN'}?
	\end{quote}
	
	
	Finally, we can give an analogous definition for \textsc{Removal $q$-Illusion Elimination}, which looks for subnetworks of \textit{SN} obtained by removing at most $k$ edges  such that an existing $q$-illusion is eliminated.
	
	In this section we will show that these problems are NP-complete for every rational $q$ in $(0,1)$ by reduction from 2P2N-SAT problem, which has been shown to be NP-complete. In 2P2N-SAT it is checked if a CNF formula in which every variable appears twice in the positive, and twice in the negative form is satisfiable (see \cite{berman2004approximation}). We will commence with showing that \textsc{$q$-Illusion Elimination} is NP-complete for every rational $q$ in $(0,1)$. 
	We begin by presenting the structures that will form our reduction, and then sketch the main lines of the proof, which can be found in complete form in the appendix.
	
	\paragraph{$k$-Pump-up gadget.}
	
	Let us construct what we call a \emph{$k$-pump-up gadget}. For a natural number $k \geq 1$ we create $k+4$ blue nodes which are not connected to each other. In addition we construct 4 red nodes, which are also not connected to each other. Furthermore, let each red node in the gadget be connected to all blue nodes in this structure.
	Observe that if a $k$-pump-up gadget is embedded in a network in which blue is the majority winner, then $k+4$ nodes are under illusion in this structure, while 4 are not. Also, for every blue node $i$ in the gadget, the margin of victory of $i$ is $-4$.

	\paragraph{$k$-Pump-down gadget.}
	
	Let us further construct what we call a \emph{$k$-pump-down gadget}. For an uneven, natural $k \geq 3$ the $k$-pump-down gadget is a $k$-clique in which blue has the majority of 1. Also, for an uneven, natural $k \geq 4$ we construct a gadget for $k-1$ and a disjoint red node.
	Observe that if a $k$-pump-down gadget is embedded in a network in which blue is the majority winner, then all $k$ members of the structure are not under illusion. Moreover, if a blue node in the gadget would be adjacent to an additional red node, then it would be pushed into illusion.

	We also need the following technical lemmas. 
	
	\begin{lemma}\label{lemma:UpExists1}
		For every pair of natural numbers $m,k>0$ and any rational number $q$ in $(0,1)$ such that $\frac{m}{k}<q$ there exists an $h$ such that $\frac{m+h}{k+h+4} < q$ but $\frac{m+h+1}{k+h+4} \geq q$.
	\end{lemma}
	
	We will further denote such a number as $h_{k,m,q}^{\#}$, or $h^{\#}$ if $k,m$ and $q$ are clear from the context.

	\begin{lemma}\label{lemma:DownExist1}
		For every rational number $q\in (0,1)$ and $m,k \in \mathbb{N}$ such that $\frac{m}{k} \geq q$ there is a natural $h$ such that $\frac{m}{k+h} < q$, but $\frac{m+1}{k+ h} \geq q$.
	\end{lemma}
	
	We denote such a number as $h^+_{m,k,q}$, or $h^+$ if $m,k$ and $q$ are clear from the context.
	We are now ready to construct the labelled social network which we will call an \emph{encoding} of a formula $\varphi$ in 2P2N form, with the set of variables $P_{\varphi} = \{p_1, \dots, p_m\}$  and the set of clauses $C_{\varphi}=\{C_1, \dots, C_n\}$. We also refer to such a network as $E_{\varphi}=(N,E,f)$.
	\paragraph{Variable, clause, and balance gadgets.}  
	Let us start with describing what we call a \emph{variable gadget}. For every variable $p_i \in P_{\varphi}$, construct two triples of nodes labelled blue, $\{p_i^1, p_i^2, p_i^3\}$ and  $\{ \neg p_i^1, \neg p_i^2, \neg p_i^3\}$. Let all literal nodes form a clique. We say that the first of them corresponds to the literal $p_i$, while the second to $\neg p_i$, and call members of these triples \emph{literal nodes}. Further, for every literal $L$ let us construct a node $A_L$, labelled blue, which we call an \emph{auxiliary node} of $L$ and let auxiliary nodes form a clique. For each literal $L$, let $A_L$ be adjacent to all literal nodes not corresponding to $L$. 
	Also, for every variable $p_i$ let us construct a node $E_i$ labelled blue, which we call the \emph{extra node} of $p_i$. Furthermore, for each variable $p_i$, let $E_i$ be adjacent to all auxiliary nodes and literal nodes not corresponding to $p_i$ or $\neg p_i$, and let all extra nodes form a clique.
	
	\begin{figure}[H]
		\centering
		
		\scalebox{0.5}{	\begin{tikzpicture}
				[->,shorten >=1pt,auto,node distance=1.2cm,
				semithick]

				\node[shape=circle,draw=black, fill=myblue] (A)  {};
				\node[shape=circle,draw=black, right of=A, fill=myblue] (B)  {};
				\node[shape=circle,draw=black, left of=A, fill=myblue] (A')  {};
				\node[draw,  fit=(A') (B), minimum height=1.2cm ](FIt1) {};
				\draw [thick,-] (A) to  (A') ;
				\draw [thick,-] (B) to  (A) ;
				\draw [thick, bend left, -] (A') to  (B) ;
				\node[shape=circle,  left of=A'] (A'')  [fontscale=4] {$p_i$};

				\node[shape=circle,draw=black, right of=B, fill=myblue] (C)  {};
				\node[shape=circle,draw=black, right of=C, fill=myblue] (D)  {};
				\node[shape=circle,draw=black, right of=D, fill=myblue] (D')  {};
				\node[draw,  fit=(C) (D'), minimum height=1.2cm ](FIt2) {};
				\draw [thick,-] (C) to  (D) ;
				\draw [thick,-] (D) to  (D') ;
				\draw [thick, bend left, -] (C) to  (D') ;
				\node[shape=circle,  right of=D'] (D'') [fontscale=4] {$\neg p_i$};
				
				\draw [thick,-] (FIt1) to  (FIt2) ;
				
				\node[shape=circle,draw=black, below of=B, fill=myblue] (N)  {};
				\node[shape=circle,  right of=N] (Q'') [fontscale=4] {$E_i$};
				
				\node[shape=circle,  below of=N] (N')  {$\dots$};
				\node[shape=circle,draw=black, left of=N',  pattern=crosshatch, pattern color=myred] (N'')  {};
				\node[shape=circle,draw=black, right of=N', pattern=crosshatch, pattern color=myred] (N''')  {};
				
				\node[draw,  fit=(N'')  (N''') ](FIt3) {};
				
				
				\node[shape=circle,draw=black, fill=myblue, above of=B] (E)  {};
				\node[shape=circle,draw=black, fill=myblue, right of=E] (F)  {};
				\draw [thick,-] (E) to  (F) ;
				\node[shape=circle,  left of=E] (Q) [fontscale=4]  {$A_{p_i}$};
				\node[shape=circle,  right of=F] (Q') [fontscale=4] {$A_{\neg p_i}$};
				
				\draw [thick,-] (F) to  (FIt1) ;
				\draw [thick,-] (FIt2) -- (E) ;
				\draw [thick,-] (FIt3) to  (N) ;

				
				
				\node[shape=circle,draw=black, above of=E, pattern=crosshatch, pattern color=myred] (G)  {};
				\node[shape=circle,  left of=G] (H)  {$\dots$};
				\node[shape=circle,draw=black, left of=H, pattern=crosshatch, pattern color=myred] (I)  {};
				\node[draw, fit=(G)  (I) ](FIt4) {};

				
				\draw [thick,-] (E) to  (FIt4) ;
				
				
				\node[shape=circle,draw=black, above of=F, pattern=crosshatch, pattern color=myred] (L)  {};
				\node[shape=circle,  right of=L] (M)  {$\dots$};
				\node[shape=circle,draw=black, right of=M, pattern=crosshatch, pattern color=myred] (q)  {};
				\node[draw,  fit=(L)  (q) ](FIt6) {};
				

				\draw [thick,-] (F) to  (FIt6) ;

		\end{tikzpicture}}
		\caption{Variable Gadget}\label{fig:VariableGadget3}
	\end{figure}
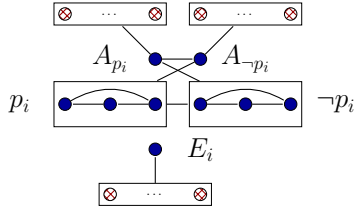

	
	In addition, let us construct what we call a \emph{clause gadget}. For every clause $C_i \in C_{\varphi}$ let us create a \emph{verifier node} $v_{C_i}$, labelled blue. We say that this node corresponds to $C_i$. Furthermore, for each clause $C_i$ and each literal $L$ not in $C_i$, let $v_{C_i}$ be adjacent to all literal nodes corresponding to $L$, as well as all auxiliary and extra nodes. Finally, for each clause $C_i$, create a group of $3  |\neg P^i| + 3  |P_{\varphi}| +1 $ nodes labelled red, where $\neg P^i$ is the set of literals which are not in $C_i$. Let all nodes corresponding to members of $\neg P^i$ be adjacent to $v_{C_i}$. Observe that in an extension of the encoding of $\varphi$ in which one additional node labelled blue is adjacent to $v_{C_i}$ (and no edges from red nodes are added to the network), illusion is eliminated from this node.

	In addition, for every auxiliary node $A_L$ create a group of nodes labelled red, adjacent to $A_L$, of the size such that there number of red nodes in the neighbourhood of $A_L$ is greater than of those labelled blue by exactly 3. Namely, let the size of such a group be $9  |P_{\varphi}|  + |C_{\varphi}| -1$. Similarly, for every literal node $p_i^j$ construct a group of nodes labelled red, adjacent to $p_i^j$ such that there is one more red node in the neighbourhood of $L_i^j$ than blue. Namely, let there be $9 |P_{\varphi}|+ \neg C^L-2$ nodes adjacent to $L_i^j$, where $C^L$ is the number of clauses in which $L$ does not appear. Also, for every extra node $E_i$, construct a group of red nodes adjacent to $E_i$ of the size such that there is one more red node in the neighbourhood of $E_i$ than the number of blue nodes in the neighbourhood of $E_i$. Namely let there be $6 |P_{\varphi}| + 9| P_{\varphi}| -6 $ such nodes.
	Finally, create a group of disconnected blue nodes of the minimal size sufficient for blue to be the strict majority in the encoding of $\varphi$. 
	
	\paragraph{Budget and requirement.} 
	We call $|P_{\varphi}|$ the \emph{requirement}, or $r_{\varphi}$. Also, we call $6|P_{\varphi}|$ the \emph{budget}, or $b_{\varphi}$. We say that network $E_{\varphi}'=(N,E',f)$ satisfies the requirement an the budget if $ | \{ e \in N^2 : e \in E \textit{ iff } e \notin E'  \} \leq  b_{\varphi}|  $ while less than $r_{\varphi}$ nodes are under illusion in $E_{\varphi}'$.

	\paragraph{Observations on modifications satisfying the budget and the requirement.}
	Let us first observe that the only nodes under illusion are literal nodes, extra nodes, auxiliary nodes and verifier nodes. Therefore, it is sufficient to eliminate the illusion from all literal, extra and verifier nodes, as well as from the half of auxiliary nodes to meet the requirement.  Furthermore, one can verify that there is no network satisfying the budget and the requirement in which some node is pushed into illusion. In addition, if no literal node is pushed into illusion, one can verify that if illusion is eliminated from more than a half of auxiliary nodes, then at least 1 extra node would remain under illusion. 
	
	For a variable gadget corresponding to $p_i$ such that in an extension \textit{SN'} of the encoding of $\varphi$ the illusion has been eliminated from $E_i$ and from one of the auxiliary nodes, we say that $p_i$ is false in \textit{SN'} if the illusion was eliminated from $A_{p_i}$, and true if it has been eliminated from $A_{\neg p_i}$. Furthermore, observe that in a network which satisfies budget and requirement, at least one auxiliary node is not under illusion in every variable gadget. Moreover, observe that by construction, in every network satisfying the budget and the requirement, exactly one edge to a blue node is added to each literal node. This entails that in every network satisfying the budget and the requirement, each verifier node $v_{C_i}$ is adjacent to a literal node corresponding to some true literal in $C_i$.

	\begin{lemma}\label{lemma:RemovalWorks}
		For every formula $\varphi$ in 2P2N form, there is a network $E'_{\varphi}=(N,E',f)$ which satisfies the requirement and the budget if and only if $\varphi$ is satisfiable.
	\end{lemma}
	
	\begin{proof}
		Take a formula $\varphi$ in 2P2N form. First, suppose it is not satisfiable. Let us further suppose, towards contradiction, that there exists a network $E'_{\varphi}$ satisfying the budget and the requirement. As $\varphi$ is not satisfiable, for every valuation $V$ over $P_{\varphi}$ there is a clause $C_i$ such that no literal in $C_i$ is true in $V$. Further, let $V$ be a valuation over $P_{\varphi}$ in which a literal $L$ is true if and only if it is true in  $E'_{\varphi}$. But this means that, by previous observations, there needs to exist a verifier node $v_{C_i}$ which is not to adjacent to a literal node corresponding to some true literal in $C_i$. So, $E'_{\varphi}$ does not satisfy the budget and the requirement which contradicts the assumptions.
		
		
		Suppose now that $\varphi$ is satisfiable. Let us construct a network $E'_{\varphi}$  satisfying the requirement and the budget. As $\varphi$ is satisfiable, there exists a valuation $V$ over $P_{\varphi}$ such that, for every clause $C_i$, there is a literal $L$ in $C_i$ which is true in $V$. Connect now edges between all literal nodes corresponding to literals false in $V$ and auxiliary nodes. Then, for every literal $L$ true in $E_{\varphi}$, construct an edge between node $L^1$  and an extra node. Finally, for every clause $C_i$, add an edge between $v_{C_i}$ and exactly one literal node corresponding to a literal true in \textit{SN}, represented in $C_i$. Note that this is always possible since $V$ is a model of $\varphi$. Notice further that $\varphi$ is in 2P2N form. Therefore, as $L$ occurs twice in $\varphi$, we can ensure that at most one edge is added between a node corresponding to $L$ and a verifier node. Finally, for every literal node $L^i$ still under illusion, add an edge between $L^i$ and any blue node in the encoding of $\varphi$. But then, in the constructed subnetwork, only $|  P_{\varphi}| $ nodes are under illusion, and edges have been added between $6 |P_{\varphi}|$ pairs of nodes. Thus,  the encoding of $\varphi$ satisfies the requirement and the budget.
	\end{proof}
	
	The above observations are sufficient to prove the main result of this section. 
	To show that \textsc{$q$-Illusion Elimination} is NP-complete for a given rational $q \in (0,1)$, we construct a labelled network $E_{\varphi}^q$ for every formula $\varphi$ in 2P2N form. 
	First, we construct $E_{\varphi}$. Let $I_{\varphi}$ denote the number of nodes under illusion in $E_{\varphi}$. Then, if $\frac{I_{\varphi} - r_{\varphi}}{|N|} < q $, we add a $ h_{|N|- r_{\varphi}, b_{\varphi}, q}^{\#}$-pump-up gadget. Otherwise, we add a $ h_{|N| - r_{\varphi}, b_{\varphi}, q}^{+}$-pump-down gadget. It follows from Lemma \ref{lemma:RemovalWorks}, as well as Lemmata \ref{lemma:DownExist1} and \ref{lemma:UpExists1} that in both cases the $q$-majority illusion can be eliminated from $E_{\varphi}^q$ if and only if $\varphi$ is satisfiable. Notice further that \textsc{$q$-Illusion Elimination } is in NP. Thus, NP-completeness of the problem follows.

	\begin{theorem}\label{theorem:removal}
		\textsc{$q$-Illusion Elimination} is NP-complete for every rational $q \in (0,1)$.
	\end{theorem}

	Furthermore, NP-completeness of \textsc{Addition $q$-Illusion Elimination} and \textsc{Removal $q$-Illusion Elimination } can be shown with a reduction similar to the one in the proof of Theorem \ref{theorem:removal}, which can be found in the appendix. 
	
	\begin{theorem}\label{theorem:ExtraRemoval}
		\textsc{Addition $q$-Illusion Elimination} and \textsc{removal $q$-Illusion Elimination} are NP-complete for every rational $q \in (0,1)$.
	\end{theorem}

	\section{Conclusions}\label{sec:conclusions}
	We have provided non-trivial constructions showing the algorithmic hardness of checking if it is possible to find a colouring of a social network in which a specified fraction of agents is under illusion, and of checking if the number of agents under illusion can be reduced to a desired level by modifying the connections between them. 
	Our research opens a number of directions for further investigations. 
	Let us mention a few particularly interesting ones.
	
	\begin{itemize}
		
		\item Establishing the complexity of checking if a network admits a $q$-majority illusion for fractions smaller than or equal to $\frac{1}{2}$ remains open (cf. Theorem~\ref{thm:verifyillusion}).
		
		\item There are social networks that do not admit a  illusion but do admit a ``plurality illusion", i.e., misperceiving the more popular option, if three or more were to be used (see the Appendix for such an example). This is particularly relevant for voting contexts such as elections with multiple candidates.
		
		\item Real-world social networks often show a high level of clustering, i.e., agents with many connections in common also tend to be connected themselves (see, e.g. \cite{fox2020finding}). 
		It is of interest to study how the ``level'' of clustering 
		can impact the existence of illusion and our complexity results.
		For example, a network where any pair of nodes with a common connection is also connected, does not admit 1- majority illusion. 

		\item Similarly, while our results show the hardness of the problems studied in the general case, identifying well-known graph parameterisations (e.g., treewidth of the network) under which they are fixed-parameter tractable is a natural direction of research that is motivated by our results. 
		Note that trees do not admit 1-majority illusion.
	\end{itemize}
		\newpage
	\bibliographystyle{named}
	\bibliography{MajIll}
	\newpage
	\section*{Appendix}
	
		\subsection*{Proof of Theorem 1}
	
	\begin{lemma*}
		\normalfont \textbf{1.} A labelling of a variable gadget (considered as a separate network) induces a 1-majority illusion only if exactly one of the nodes in the bottom pair is labelled $r$. 
	\end{lemma*}
	\begin{proof}
		Observe that a labelling of this gadget (considered as a separate network) induces a 1-majority illusion only if exactly one of the nodes in the bottom pair is labelled $r$. 
		To see that this observation holds, notice that all nodes in the 4-clique need to be labelled $r$ for the 1-majority illusion to hold, as all of them have at least one dependant.
		Further, as there are eleven nodes in the gadget, only five of them can be labelled $r$ for a labelling to induce the 1-majority illusion in such a structure (considered as a separate network). But 1-majority illusion will not be induced for this structure if nodes $p_i$ and $\neg p_i$ situated at the bottom of the variable gadget  are both labelled $b$. {Indeed, if, in some labelling of the gadget which induces 1-majority illusion, $p_i$ and $\neg p_i$ were labelled $b$, then at least one of the nodes in the 4-clique would not be under illusion, which is not possible. This holds as, by the definition of 1-majority illusion, all of the nodes in the clique are labelled with $r$ and at least six nodes in the gadget are labelled $r$, in a labelling of a variable gadget which induces illusion. It implies that one of the nodes in the clique would be linked to three red nodes and three blue nodes.} %
	\end{proof}

	\begin{lemma*}
		\normalfont \textbf{2.} There exists a labelling of a clause gadget (not as a separate network) which induces 1-majority illusion with blue being a majority winner in this structure if and only if at least one node is linked to a literal node labelled $r$. 
	\end{lemma*}
	
	\begin{proof}
		Suppose that there is such a node in the clause gadget, call it $L$. {Then a labelling inducing 1-majority illusion in the gadget can be constructed as in Figure 3, where $L$ is the left-top member of the clique.}
		On the contrary, if all literal nodes linked to members of a clause gadget were labelled $b$, the illusion could not hold, as then for all members of the gadget linked to literal nodes we would need to have at least one dependant or co-dependant labelled $r$. Otherwise, we would have at least as many red as blue members of these nodes' respective neighbourhoods. But then, as in a labelling of the gadget inducing 1-majority illusion we would also need the 5-clique to be labelled $r$, in such a labelling there would be at least as many red as blue nodes in the gadget. 
	\end{proof}

	\begin{lemma*}
		\normalfont \textbf{3.} In a labelling of $E_{\varphi}$ which induces 1-majority illusion every variable gadget is  of type A or type B.
	\end{lemma*}
	\begin{proof}
		Indeed, in such a gadget at least one literal node needs to be labelled red for the bottom nodes in the 4-clique to be under illusion. On the other hand, one of them needs to be blue, as otherwise blue would not be the majority colour in $E_{\varphi}$. 
	\end{proof}
	
	\begin{lemma*}\label{lemma:RigidEncoding}
		\normalfont \textbf{5.} For every 3-CNF formula $\varphi$, $k \leq I_{\varphi}$ and any labelling $f$ of $E_{\varphi}$ such that $R_f = I_{\varphi} - k$, the number of nodes under illusion in $E_{\varphi}$ under $f$ is at most $| N| - k $.
	\end{lemma*}

	\begin{proof}
		Consider a 3-CNF formula $\varphi$ and $k \leq I_{\varphi}$, as well as a labelling $f$ of $E_{\varphi}$ such $R_f = I_{\varphi} - k$. We will show that the number of nodes under illusion in $E_{\varphi}$ under $f$ is at most $| N| - k $.		
		Let us denote as $A$ the set of all 4-cliques variable gadgets of $E_{\varphi}$. Moreover, let $B'$ be the set of all members of 5-cliques in clause gadgets. Also, let $C$ be the set of all literal nodes, and $D$ be the set of all co-dependants in clause gadgets. Further, let $E$ be the set of all nodes in the balance gadget, and $F$ be the set of all other nodes. Observe that each node in $F$ is a dependant of some node outside of $F$. 		
		Let us state some initial observations. Firstly, notice that $E_{\varphi} = A \cup B' \cup C \cup D \cup E \cup F$. Moreover, by construction of $E_{\varphi}$, $|A|+|B'|+ \frac{|C|}{2} + \frac{2|D|}{3} + |E| = I_{\varphi}$. Further, we will present crucial properties of sets $A$, $B'$, $C$, $D$ and $E$.
			Observe that as each node $i \in A$ has a dependant, then there is a set $N_A \subseteq F$ with $|N_A| = |B^A |$ such that for every $i \in N_A$, $i$ is not under illusion. Similarly, as each node $i \in B'$ has a dependant, then there is a set $N_B' \subseteq F$ with $|N_B'| = |B^{B'} |$ such that for every $i \in N_B'$, $i$ is not under illusion. 
			Further, Let  $M_C= |B^C| - \frac{|C|}{2} $ if $|B^C| - \frac{|C|}{2} >0 $ and let $M_C=0$ otherwise. Observe that as $|C|$ is even, $M_C$ is a natural number. Also, let $I_C$ be the set of members of $A$ which are linked to 2 blue literal nodes but are under illusion. Observe further that for a variable gadget corresponding to some $p_i$, if both literal nodes in the gadget are labelled blue, then each of the 2 members of $A$ linked to them is under illusion only if its dependant is labelled red. Thus, $2 \cdot M_c \geq |I_{\varphi}|$.
			Also, notice that for every triple $T$ of co-dependants in a clause gadget corresponding to some clause $C_i$, if less than 2 members of $T$ are labelled red, then all members of $T$ are not illusion. Moreover, if 2 of them are labelled red, then all nodes in $T$ are under illusion. Then, let $M_D=|B^D| - \frac{2 |D|}{3} $ if $|B^D| > \frac{2 |D|}{3}$, and let $M_D=0$ otherwise. Note that $|D|$ is divisible by 3, so $M_D$ is a natural number. Furthermore observe, that the number of nodes not under illusion in $D$ is at least $\frac{3}{2}  \cdot M_D $.
			Let us also observe that given the set $B^E$, the number of nodes not under illusion in $E$ is greater or equal to $|B^E|$. 
		Furthermore, observe that as  $|A|+|B'|+ \frac{|C|}{2} + \frac{2|D|}{3} + |E| = I_{\varphi}$ and at most $ I_{\varphi}$ nodes are labelled red in $E_{\varphi}$. Also, as for every $i \in I_C$ there is a $j \in F$ labelled red, we have that $ |B^A|+|B^{B'}|+ |B^C| + |B^D| + |B^E| \geq k +\frac{|C|}{2} + \frac{|D|}{3} +|I_C| $. Also, notice that by properties of $A$, $B'$ and $E$, we have that the number of nodes not under illusion in $E_{\varphi}$ is at least $|B^A|+|B^{B'}|+|B^E|$. Further, there are at least $\frac{3}{2}  \cdot M_D $ additional nodes under illusion. Lastly, there are at least $2 \cdot M_C - | I_C |$ additional nodes under illusion which, as we observed before, is greater than 0. Let us now show that there are at least $k$ nodes not under illusion in $E_{\varphi}$. Suppose that $M_C \geq | I_C |$. Then, the claim follows immediately, as we know that $|B^A|+|B^{B'}|+|B^E|+ M_C +M_D \geq k$. Also, if $M_C < | I_C |$, then $|B^A|+|B^{B'}|+|B^E|+M_D >k$, which entails that the number of nodes not under illusion in $E_{\varphi}$ is at least $k$. 
	\end{proof}
	
	\begin{lemma*}\label{lemma:h-exists}
		\normalfont \textbf{6.} Let $q$ be a rational number in  $(\frac{1}{2}, 1]$, and $k>0$ be a natural number. Then, there exists a natural number $h^*$ such that $\frac{k+ h^*}{k+ 2h^*  }  \geq q$, but $\frac{k+ h^* -1}{k+ 2h^*  } < q $. 
	\end{lemma*}
	
	\begin{proof}
		Take a $k \in \mathbb{N}_+$ and a fraction $q \in \mathbb{Q} \cap (\frac{1}{2},1  ]$. Observe that if $\frac{a}{b}=1$, then the claim holds immediately. So, we will only consider fractions such that $\frac{a}{b}<1$. Then, we define a function $f: \mathbb{N} \rightarrow \mathbb{Q}$ such that for a natural $h$, $f(h)= \frac{k+h}{k+2h}$. Observe first that $f(0)=1$. Also, $f$ is strictly downwards monotone, and is bounded by $\frac{1}{2}$. But then, as $q \in (\frac{1}{2},1  ]$, there needs to exist a maximal $h$ such that $f(h) \geq q$, and as $f$ is strictly downwards monotone, $f(h+1) < q$. We denote such a number as $h^*$.
		Notice now that (1): $ \frac{k+h^*}{k+2h^*} \geq \frac{a}{b}$ by definition of $h^*$. Note further that if $\frac{k+h^*}{k+2h^*} = \frac{a}{b}$, then the claim holds immediately. Let us assume then that $\frac{k+h^*}{k+2h^*} < \frac{a}{b}$.  Further, suppose towards contradiction that (2): $\frac{k+ h^* -1}{k+ 2h^*  } \geq \frac{a}{b} $. Also observe that  (3) $ \frac{k+ h^* +1}{k+ 2h^* +2 } < \frac{a}{b} $.  Now, from (1) we get that $ b(k+h^*) \geq a(k + 2h^*)  $. Also, from (2) we get that $b(k+h^*-1 ) \geq a(k+2h^*)$, which is equivalent to $ bk+bh^*-b \geq ak +2ah^*  $, and also to $-bk -bh^* +b \leq -ak -2ah^*$. We denote this inequality as $\alpha$. Also, from (3) we have $b(k+h^*+1) <a(k+2h^*+2)   $, which is equivalent to $ bk + bh^* +b < ak +2ah^* +2a$. We denote this inequality as $\beta$ By adding $\alpha$ and $\beta$ we get that $ 2b \leq 2a$, which is impossible since $a<b$.
	\end{proof}
	
	\begin{theorem*} \normalfont \textbf{1. }
		\textsc{$q$-majority illusion} is NP-complete for every rational $q$ in  $(\frac{1}{2},1]$.
	\end{theorem*}
	
	\begin{proof}
		Take any rational $q$ in $(\frac{1}{2},1]$. First, observe that \textsc{$q$-majority illusion} is in NP. We will now show that it is NP-hard by reduction from \textit{3-SAT}. 
		
		Consider a 3-CNF formula $\varphi$ with the set $P_{\varphi} = \{p_1, \dots, p_n \} $ of variables and the set $C_{\varphi}=\{C_1, \dots, C_m\}$ of clauses. Let us construct what we call a \emph{$q$-encoding} of $\varphi$. First, let $E_{\varphi}$ be a subnetwork of the $q$-encoding of $\varphi$. Moreover, construct $h_{|E_{\varphi}|, q}^*$ pairs of nodes, such that nodes in each such pair are connected to each other, but not to any other node in the network. We call this set of pairs $H$. Observe further that the $q$-encoding of $\varphi$ can be constructed in polynomial time. Also, by Lemma 6, the $q$-encoding of $\varphi$ is a $q$-majority illusion for some labelling $f$ if at least $|E_{\varphi}| + h^*_{|E_{\varphi}|, q}$ nodes are under illusion in $f$.
		
		Let us show that the $q$-encoding of $\varphi$ admits $q$-majority illusion if and only if $\varphi$ is satisfiable. First, suppose that $\varphi$ is satisfiable. Then observe that as $\varphi$ is satisfiable, by Lemma 4 $E_{\varphi}$ admits 1-majority illusion as a separate network. Hence, there is a labelling of the $q$-encoding of $\varphi$ such that exactly $I_{\varphi}$ nodes in $E_{\varphi}$, as well as one of nodes in each additional pairs are labelled red, and $|E_{\varphi}| +  h^*_{|E_{\varphi}|, q}$ nodes are under illusion. Hence, the $q$-encoding of $\varphi$ admits $q$-majority illusion.
		
		Suppose now that $\varphi$ is not satisfiable. Then, suppose that there is a labelling $f$ of the $q$-encoding of $\varphi$ which induces $q$-majority illusion. Let us first observe that if less than $h^*_{|E_{\varphi}|, q}$ are labelled red in $H$, then $f$ does not induce $q$-majority illusion. Indeed, if it was the case, then less than $h^*_{|E_{\varphi}|, q}$ nodes in $H$ would be under illusion, and hence the number of nodes under illusion in the $q$-encoding of $\varphi$ would be strictly smaller than $|E_{\varphi}| + h^*_{|E_{\varphi}|, q}$. But then, as $f$ induces $q$-majority illusion, at least $h^*_{|E_{\varphi}|, q}$ are labelled red in $H$. So, the number of nodes labelled red in $E_{\varphi}$ is smaller or equal to $I_{\varphi}$. If it is equal to $I_{\varphi}$, then the number of nodes under illusion in $H$ is  $h^*_{|E_{\varphi}|, q}$, but as $\varphi$ is not satisfiable, not all members of $E_{\varphi}$ are under illusion, and hence $f$ does not induce $q$-majority illusion. Now, suppose that less than $I_{\varphi}$ nodes are labelled red in $|\varphi|$. Let $k = |I_{\varphi} - R_{\varphi}|$. Further, let us denote as $M$ the maximum number of nodes under illusion in $E_{\varphi}$ if $I_{\varphi}$ nodes are labelled red in this subnetwork. Now, by Lemma \ref{lemma:RigidEncoding} we have that the number of nodes under illusion is at most $M-k$. But then, the number of nodes labelled red in $H$ is at most $h^*_{|E_{\varphi}|, q} +k$, and hence the number of nodes under illusion in the $q$-encoding of $\varphi$ is at most $ M-k + h^*_{|E_{\varphi}|,q} $, which is smaller than $|E_{\varphi}| +  h^*_{|E_{\varphi}|,q}$ since $M < |E_{\varphi}|$.

	\end{proof}
	
	\subsection*{Proof of Theorem 2}
	
	\begin{lemma*}\label{lemma:UpExists1}
		\normalfont \textbf{7.} For every pair of natural numbers $m,k>0$ and any rational number $q$ in $(0,1)$ such that $\frac{m}{k}<q$ there exists an $h$ such that $\frac{m+h}{k+h+4} < q$ but $\frac{m+h+1}{k+h+4} \geq q$.
	\end{lemma*}
	
	\begin{proof}
		Take any such $k,m$ and $q$ , and let $q=\frac{a}{b}$. We define a function $f_{m,k,q}: \mathbb{N} \rightarrow \mathbb{Q}$ such that for a natural number $h$, $f_{k,m,q}=  \frac{m+h}{k+h+4}$. Firstly observe that as $\frac{m}{k}<q$ it holds that $f_{m,k,q}(0)<q$. Moreover, observe that $f_{k,m,q}$ is strictly increasing and bounded by 1. Therefore, there exists an $h \in \mathbb{N}$ such that $f_{m,k,q}(h) < q$ while $f_{m,k,q}(h+1) \geq q$. We call such a number $h^{\#}$.
		Then, suppose towards contradiction that $\frac{m+h^{\#}+1}{k+h^{\#}+4} < \frac{a}{b}$. Then, we have that $b(m+h^{\#}+1)<a(k+h^{\#} +4)$, which is equivalent to $a(m+h^{\#})+4a > b(m+h^{\#})+b$. We denote this inequality as $\alpha$. Additionally, as $f_{k,m,q}(h^{\#}+1) \geq q$, we know that $\frac{m+h^{\#}+1}{k+h^{\#}+5} \geq \frac{a}{b}  $. So, $a(k+h^{\#}+5) \leq b(m+h^{\#}+1)$, and thus $-a(k+h^{\#}+5) \geq -b(m+h^{\#}+1)$. This is equivalent to $ -a(k+h^{\#})-5a \geq -b(m+h^{\#}+1) $. We denote this inequality as $\beta$. By adding $\alpha$ and $\beta$ we get that $-a \geq 0$, so $a \leq 0$. But this is impossible since $\frac{a}{b} > 0$.

	\end{proof}

	\begin{lemma*}\label{lemma:DownExist1}
		\normalfont \textbf{8.} For every rational number $q\in (0,1)$ and $m,k \in \mathbb{N}$ such that $\frac{m}{k} \geq q$ there is a natural $h$ such that $\frac{m}{k+h} < q$, but $\frac{m+1}{k+ h} \geq q$.
	\end{lemma*}

	\begin{proof}
		Take any such $m,k $ and $q$, and let $q=\frac{a}{b}$. We define a function $g_{m,k,q}: \mathbb{N} \rightarrow \mathbb{Q}$ such that for a natural number $h$, $f_{k,m,q}=  \frac{m}{k+h}$. Observe that $g_{m,k,q}(0) = \frac{m}{k}$ and that $g_{m,k,q}$ is strictly decreasing and bounded by 0. So, there exists a natural $h$ such that $g_{m,k,q}(h) < q$ but $g_{m,k,q}(h-1) \geq q $, as $q>0$. We will further call  $h^+_{m,k,q}$.
		Then, suppose towards contradiction that $\frac{m+1}{k+ h^+_{m,k,q}} < \frac{a}{b}$. We will further call  $h^+_{m,k,q}$ as $h^+$ for brevity. Then, we have that $bm+b<ak + ah^+$, and so $ -bm-b > -ak - ah^+$. We denote this inequality as $\alpha$. Also, notice that by definition of $h^+$ we get that $\frac{m}{k+h^+-1} \geq \frac{a}{b}$. So, $bm \geq ak+ah^+-a$. We denote this inequality as $\beta$. By adding $\alpha$ and $\beta$ we get that $-b \geq -a$, and so $a \geq b$ which is impossible since $\frac{a}{b}<1$.
	\end{proof}
	
	\begin{theorem*} \textbf{1.}
		\textsc{$q$-Illusion Elimination } is NP-complete for every rational $q \in (0,1)$.
	\end{theorem*}

	\begin{proof}
		
		Consider any rational $q \in (0,1)$. First, observe that \textsc{$q$-Illusion Elimination} is in NP. Let us further construct a network $E^q_{\varphi}$ for a 2P2N formula $\varphi$.
		The first component of $E^q_{\varphi}$ is $E_{\varphi}$. If $\frac{I_{\varphi} - r_{\varphi}}{|E_{\varphi}|} < q $, then  construct a pump up gadget for $k= h_{|N|- r_{\varphi}, b_{\varphi}, q}^{\#}$. Otherwise, construct a pump down gadget for $k'= h_{|N| - r_{\varphi}, b_{\varphi}, q}^{+}$. Let us now show that the answer to  \textsc{$q$-Illusion Elimination} for $E^q_{\varphi}$ and $b_{\varphi}$ is positive if and only if $\varphi$ is satisfiable.
		
		$(\Rightarrow)$ Suppose that $\varphi$ is satisfiable. We will show that the answer to $E^q_{\varphi}$ and $b_{\varphi}$ is positive. We denote as $|E^q_{\varphi}|$ the number of nodes in $E^q_{\varphi}$.  Let us first consider the case in which $\frac{I_{\varphi} - r_{\varphi}}{|E^q_{\varphi}|} < q $. Then observe that as $\varphi$ is satisfiable, by Lemma 6 we have that it is possible to find a subnetwork $E'_{\varphi}$ of $E_{\varphi}$ in which $b_{\varphi}$ edges are eliminated, and where illusion was eliminated from $r_{\varphi}$ nodes. But then, by Lemma 7 we get that $\frac {I_{\varphi} - r_{\varphi} + k'} {I_{\varphi} - r_{\varphi} + k' +4} <q $. So we can construct a netnetwork of $E^q_{\varphi}$ in which only $b_{\varphi}$ edges altered but $q$-majority illusion does not hold.
		Similarly, if   $\frac{I_{\varphi} - r_{\varphi}}{|E^q_{\varphi}|} \geq q $ we observe that by Lemma 8 we get that  $\frac {I_{\varphi} - r_{\varphi}}  {|E^q_{\varphi} | + k} < q$. So, we get that as we can eliminate illusion from $r_{\varphi}$ nodes in  $E_{\varphi}$ by modifying $b_{\varphi}$ edges. But then we can construct a subnetwork of $E^q_{\varphi}$ in which only $b_{\varphi}$ edges are removed but $q$-majority illusion does not hold.
		
		$(\Leftarrow)$ Suppose now that $\varphi$ is not satisfiable. We will show that the answer to the considered problem for $E^q_{\varphi}$ and $b_{\varphi}$ is negative. let us first consider the case in which $\frac{I_{\varphi} - r_{\varphi}}{|E^q_{\varphi}|} < q $. Notice that by Lemma 7  that the minimum number of nodes from which illusion needs to be removed for $q$-majority illusion not to hold in $E^q_{\varphi}$ is $r_{\varphi}$. Furthermore observe that in the pump up gadget, the minimum number of edges which is needed to be added to eliminate the illusion from a single node is greater than 4. Moreover, there is a set of nodes $S$ under illusion in $E_{\varphi}$ such that for every $i \in S$, illusion can be eliminated from $i$ by adding 3 edges, without pushing any node under illusion. Further, as $\varphi$ is not satisfiable, by Lemma 6 we get that at it is not possible to remove the illusion from at least $r_{\varphi}$ nodes in $E_{\varphi}$. But then, it is also not possible to remove the illusion from at least $r_{\varphi}$ in $E^q_{\varphi}$	The reasoning for the case in which $\frac{I_{\varphi} - r_{\varphi}}{|E^q_{\varphi}|} \geq q $ is symmetric.

	\end{proof}

	\subsection*{Proof of Theorem 3}

	Let us show that \textsc{Addition $q$-Illusion Elimination } is NP-complete for every rational $q \in (0,1)$. To do that, we will provide a reduction of \textsc{Addition $q$-Illusion Elimination } from 2P2N-SAT for every such $q$.
	
	Let us construct the labelled social network which we will call an \emph{encoding} of a formula $\varphi$ in 2P2N form, or $E_{\varphi} = (N,E,f)$. We assume that $\varphi$ has  a set of variables $P_{\varphi}=\{p_1, \dots, p_m\}$ and a set of clauses $C_{\varphi} = \{C_1, \dots, C_n   \}$.
	
	
	\paragraph{Variable, clause and balance gadgets}
	Let us describe what we call a \emph{variable gadget}. For every variable $p_i \in P_{\varphi}$ construct two triples of nodes labelled blue, $\{p_i^1, p_i^2, p_i^3\}$ and  $\{\neg p_i^1, \neg p_i^2, \neg p_i^3\}$. We say that the first of them corresponds to the literal $p_i$, while the second to the literal $\neg p_i$. We call nodes in such triples \emph{literal nodes}. Further, for the literal triple corresponding to a literal $L$, construct a node $A_L$, which we call an \emph{auxiliary node}. Let $A_L$ be linked to all literal nodes apart from the triple corresponding to $L$. Furthermore, for every auxiliary node $A_L$ construct a group of red nodes $R_L$ of size $6|P_{\varphi}| -3$.
	and link them to $A_L$. Notice that by construction the illusion is eliminated from $A_L$ if it is linked to three additional blue nodes and that it is already linked to all literal nodes not corresponding to $L$. Also, for a variable $p_i$, construct a red node $E_{i}$, which we call an \emph{extra node}. Let it be linked to all literal nodes apart from the nodes corresponding to $p_i$ or $\neg p_i$. Furthermore, for every node $E_{p_i}$ construct a group of red nodes $R^E_{i}$ of size $6|P_{\varphi}| -5$. Notice that by construction of the gadget, the illusion is eliminated from $E_{i}$ if it is linked to one additional blue node and that it is already linked to all literal nodes not corresponding to $p_i$ or to $\neg p_i$.
	
	
	The construction for a variable $p_i$ is depicted in Figure 5.
	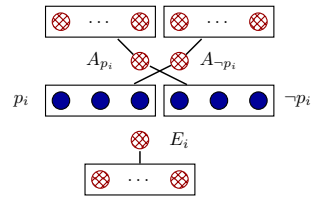
\begin{figure}[H]
		\centering
		
		\scalebox{0.7}{	\begin{tikzpicture}
				[->,shorten >=1pt,auto,node distance=0.75cm,
				semithick]

				\node[shape=circle,draw=black, fill=myblue] (A)  {};
				\node[shape=circle,draw=black, right of=A, fill=myblue] (B)  {};
				\node[shape=circle,draw=black, left of=A, fill=myblue] (A')  {};
				\node[draw,  fit=(A') (B) ](FIt1) {};
				\node[shape=circle,  left of=A'] (A'')  {$p_i$};

				\node[shape=circle,draw=black, right of=B, fill=myblue] (C)  {};
				\node[shape=circle,draw=black, right of=C, fill=myblue] (D)  {};
				\node[shape=circle,draw=black, right of=D, fill=myblue] (D')  {};
				\node[draw,  fit=(C) (D') ](FIt2) {};
				\node[shape=circle,  right of=D'] (D'')  {$\neg p_i$};

				\node[shape=circle,draw=myred, pattern=crosshatch, pattern color=myred, below of=B] (N)  {};
				\node[shape=circle, right of=N] (Z)  {$E_i$};
				
				\node[shape=circle,draw=myred, pattern=crosshatch, pattern color=myred, below of=N, left of=N] (N')  {};
				\node[shape=circle,  right of=N'] (N'')  {$\dots$};
				\node[shape=circle,draw=myred, pattern=crosshatch, pattern color=myred, right of=N''] (N''')  {};
				\node[draw,  fit=(N')  (N''') ](FIt3) {};
				
				
				\node[shape=circle,draw=myred, pattern=crosshatch, pattern color=myred, above of=B] (E)  {};
				\node[shape=circle,draw=myred, pattern=crosshatch, pattern color=myred, right of=E] (F)  {};
				
				\draw [thick,-] (F) to  (FIt1) ;
				\draw [thick,-] (FIt2) -- (E) ;
				\draw [thick,-] (FIt3) to  (N) ;

				
				
				\node[shape=circle,draw=myred, pattern=crosshatch, pattern color=myred, above of=E] (G)  {};
				\node[shape=circle,  left of=G] (H)  {$\dots$};
				\node[shape=circle,draw=myred, pattern=crosshatch, pattern color=myred, left of=H] (I)  {};
				\node[draw, fit=(G)  (I) ](FIt4) {};
				
				
				\draw [thick,-] (E) to  (FIt4) ;
				\node[shape=circle, left of=E](Z1) {$A_{p_i}$};
				
				\node[shape=circle,draw=myred, pattern=crosshatch, pattern color=myred, above of=F] (L)  {};
				\node[shape=circle,  right of=L] (M)  {$\dots$};
				\node[shape=circle,draw=myred, pattern=crosshatch, pattern color=myred, right of=M] (q)  {};
				\node[draw,  fit=(L)  (q) ](FIt6) {};
				

				\draw [thick,-] (F) to  (FIt6) ;
				\node[ right of =F](Z2) {$A_{\neg p_i}$};
				
		\end{tikzpicture}}
		\caption{Variable Gadget}\label{fig:VariableGadget3}
	\end{figure}
	
	Furthermore, let us define what we call a \emph{clause gadget}. For every clause $C_i \in C_{\varphi}$ construct what we call a \emph{verifier} node $v_{C_i}$. Further, let $L_{i}$ denote the number of literals in $C_i$. Construct a group of $6n- 3L_{i} +1$ red nodes and link them to $v_{C_i}$. Furthermore, let $v_{C_i}$ be linked to all literal nodes corresponding to literals not represented in $C_i$. Notice that $v_{C_i}$ is linked to $6n- 3L_{i}$ blue nodes. Hence, it is sufficient to link $v_{C_i}$ to one additional blue node to eliminate illusion from it, but it cannot be achieved by linking it to a literal node corresponding to a literal outside of $C_i$. 
	
	The construction of a  clause gadget is depicted in Figure 6.
	\begin{figure}[H]
		\centering
		
		\scalebox{0.5}{\begin{tikzpicture}
				[->,shorten >=1pt,auto,node distance=1.2cm,
				semithick]
				\node[shape=circle,draw=myred, pattern=crosshatch, pattern color=myred] (A)  {};
				\node[shape=circle,  right of=A] (Q) [fontscale=4] {$v_{C_k}$};
				\node[shape=circle,draw=myred, pattern=crosshatch, pattern color=myred, below of= A, left of=A] (B)  {};
				\node[shape=circle,draw=myred, pattern=crosshatch, pattern color=myred, below of=A] (C)  {};
				\node[shape=circle,  right of=C] (D)  {$\dots$};
				\node[shape=circle,draw=myred, pattern=crosshatch, pattern color=myred, right of=D] (E)  {};
				
				\draw [thick,-] (A) to  (B) ;
				\draw [thick,-] (A) to  (C) ;
				\draw [thick,-] (A) to  (D) ;
				\draw [thick,-] (A) to  (E) ;
				\node[draw,  dashed, fit=(A) (C)  (B) (C)(E) ](FIt1) {};
				
				\node[shape=circle,draw=black, fill=myblue, above of= A] (A')  {};
				\node[shape=circle,  right of=A'] (B')  {$\dots$};
				\node[shape=circle,draw=black, fill=myblue, right of= B'] (C')  {};
				
				\draw [thick,-] (A) to  (A') ;
				\draw [thick,-] (A) to  (B') ;
				\draw [thick,-] (A) to  (C') ;

		\end{tikzpicture}}
		\caption{Clause Gadget.}\label{fig:clausegadget3}
	\end{figure}
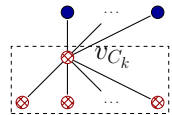
	Also, for every triple corresponding to a literal $L$, create what we call a \emph{balance gadget}. Let $B_{L}$ denote the number of clauses that $L$ is not in. 
	Then, for a node $v$ in the triple corresponding to $L$ construct and link to $v$ 
	a minimal number blue nodes such that after linking them to $v$, blue has majority of 2 in the neighbourhood of $v$. Namely, this number amounts to $B_L + 3|P_{\varphi}| -1$. Also, let each of these blue nodes be linked to a newly constructed red node. Observe that the newly connected blue nodes have the same number of blue and red neighbours.
	
	A balance gadget is depicted in Figure 7.
	\begin{figure}[H]
		\centering
		
		\scalebox{0.6}{	\begin{tikzpicture}
				[->,shorten >=1pt,auto,node distance=.8cm,
				semithick]

				\node[shape=circle,draw=black, fill=myblue] (A)  {};
				\node[shape=circle,draw=black, right of=A, fill=myblue] (B)  {};
				\node[shape=circle,draw=black, right of=B, fill=myblue] (Q)  {};
				\node[draw,  fit=(A) (Q) ](FIt1) {};
				\node[right of = FIt1 ](Z) {$L$};

				\node[shape=circle,draw=black, above of=A, fill=myblue] (C)  {};
				\node[shape=circle,  left of=C] (D)  {$\dots$};
				\node[shape=circle,draw=black, left of=D, fill=myblue] (E)  {};
				
				\node[shape=circle,draw=myred, pattern=crosshatch, pattern color=myred, above of=C] (F)  {};
				\node[shape=circle,   left of=F] (G)  {$\dots$};
				\node[shape=circle,draw=myred, pattern=crosshatch, pattern color=myred, left of=G] (H)  {};
				
				\draw [thick,-] (C) to  (F) ;
				\draw [thick,-] (D) to  (G) ;
				\draw [thick,-] (E) to  (H) ;
				
				\draw [thick,-] (C) to  (A) ;
				\draw [thick,-] (E) to  (A) ;

				\node[shape=circle,draw=black, above of=Q, fill=myblue] (I)  {};
				\node[shape=circle,  right of=I] (J)  {$\dots$};
				\node[shape=circle,draw=black, right of=J, fill=myblue] (K)  {};
				
				\node[shape=circle,draw=myred, pattern=crosshatch, pattern color=myred, above of=I] (L)  {};
				\node[shape=circle,   above of=J] (M)  {$\dots$};
				\node[shape=circle,draw=myred, pattern=crosshatch, pattern color=myred, above of=K] (N)  {};
				
				\draw [thick,-] (I) to  (L) ;
				\draw [thick,-] (J) to  (M) ;
				\draw [thick,-] (K) to  (N) ;
				
				\draw [thick,-] (I) to  (Q) ;
				\draw [thick,-] (K) to  (Q) ;
				
				\node[shape=circle,draw=black, below of=Q, fill=myblue] (R)  {};
				\node[shape=circle,  left of=R] (R')  {$\dots$};
				\node[shape=circle,draw=black, left of=R', fill=myblue] (R'')  {};
				
				\node[shape=circle,draw=myred, pattern=crosshatch, pattern color=myred, below of=R] (R''')  {};
				\node[shape=circle,   left of=R'''] (R'''')  {$\dots$};
				\node[shape=circle,draw=myred, pattern=crosshatch, pattern color=myred, left of=R''''] (R''''')  {};
				
				\draw [thick,-] (B) to  (R) ;
				\draw [thick,-] (B) to  (R'') ;
				
				\draw [thick,-] (R'') to  (R''''') ;
				\draw [thick,-] (R) to  (R''') ;

		\end{tikzpicture}}
		\caption{Balance gadget for a literal $L$.}\label{fig: balance3}
	\end{figure}
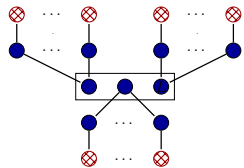
	
	Finally, connect all constructed red nodes other than verifier, auxiliary and extra nodes into a clique. Also, add the minimal number of disconnected blue nodes sufficient for blue to be the strict majority colour in the encoding.
	\paragraph{Budget and requirement}
	
	Let $I_{\varphi}$ denote the total number of nodes which are under illusion in the encoding of $\varphi$. In the constructed instance of \textsc{Bounded Illusion Elimination} we require that in a constructed social network $\textit{SN'}=(N,E,f)$ extending the encoding of $\varphi$,  $|\{n \in N: W^n_{\it SN'} =r \}| \leq I_{\varphi} - |C_{\varphi} | - 2 |P_{\varphi}|$ and that connections between not more than $|C_{\varphi}|+4|P_{\varphi}|$ pairs of nodes are added to the encoding. We call $ |C_{\varphi} | + 2 |P_{\varphi}|$ the \emph{requirement} and denote it as $r_{\varphi}$. We further call $|C_{\varphi}|+4|P_{\varphi}|$ the \emph{budget} and denote it as $b_{\varphi}$.  We say that such a network \textit{SN'} satisfies the budget and the requirement.
	
	\paragraph{Observations on networks satisfying the budget and requirement} 
	
	Let us show that in any such extension of the encoding of $\varphi$, illusion is eliminated from all clause nodes, extra nodes and a half of the auxiliary nodes. To see that it is enough to notice that, by construction, these nodes maximize the sum of margins of victory among the sets of nodes under illusion of size  $|C_{\varphi}| + 2 |P_{\varphi}|$, and that this sum is equal to $|C_{\varphi}|+4|P_{\varphi}|$. As a consequence of this observation we get that in any network extending the encoding of $\varphi$ satisfying the budget and requirement, no node is pushed into illusion.  
	
	Let us also observe that in a network extending the encoding of $\varphi$ satisfying the budget and requirement, for a variable $p_i$, illusion is eliminated from exactly one of $A_{p_i}$ and $A_{\neg p_i}$. Notice that as a consequence of this observation we get that the choice of auxiliary nodes from which illusion is eliminated corresponds to a valuation over $P_{\varphi}$, in which a literal $L$ is true if illusion has not been eliminated from $A_L$, and false otherwise.
	
	\begin{lemma*}
		\normalfont \textbf{10.} For every formula $\varphi$ in 2P2N form, $E_{\varphi}$ is a subnetwork of some network $E'_{\varphi}$ including at most $k$ edges not present in $E_{\varphi}$ and less than $b_{\varphi}$ are under illusion in $E'_{\varphi}$ if and only if $\varphi$ is satisfiable.
	\end{lemma*}
	
	\begin{proof}
		
		Take a formula $\varphi$ in 2P2N form.  Let us show that the answer to \textsc{Addition $q$Illusion Elimination}  given the encoding of $\varphi$ and specified requirement and budget is positive if and only if $\varphi$ is satisfiable. 
		
		$(\Rightarrow)$  Suppose that $\varphi$ is not satisfiable. Suppose further that there exists a network $\textit{SN}$ extending the encoding of $\varphi$ which satisfies the budget and requirement. Recall that the choice of auxiliary nodes from which illusion is eliminated corresponds to a valuation $V$ over $P_{\varphi}$. Further, since $\varphi$ is not satisfiable, $V$ is not a model of $\varphi$.  Therefore, there must exist a clause node $v_{c_i}$ such that for every literal $L$ in $C_i$, $L$ is false in $V$. So, all nodes corresponding to $L$ have been linked to $A_L$. Hence, by construction no blue node can be linked to $v_{c_i}$ without being pushed into illusion. Hence, illusion cannot be eliminated from  $v_{c_i}$ without some node being pushed into illusion. Hence, $\textit{SN}$ cannot satisfy the budget and the requirement.
		
		$(\Leftarrow)$ Suppose now that $\varphi$ is satisfiable. Let us show that it is possible to remove illusion from $|C |+ 2 |V|$ nodes without any node being pushed into illusion by adding at most $|C|+4|V|$ edges to the encoding. Take a model $M$ of $\varphi$, which exists since $\varphi$ is satisfiable. Notice that for every clause $C_i$ there is a literal $L$ in $C_i$ which is true in $M$. Denote such a literal as $L_i$. Observe that since $\varphi$ is in 2P2N, for every literal $L$ there are at most two clauses $C_j, C_k$ such that $L_j=L_k=L$. Let us now construct an extension of the encoding of $\varphi$ satisfying the budget and requirement. Link every clause node $v_{c_k}$ to a node corresponding to $L_k$, with no literal node having an added link to more than one clause node. Note that this is possible since every literal $L$ only appears twice in $\varphi$ and there are three nodes corresponding to $L$. Observe that, by construction, illusion is eliminated from $|C_{\varphi}|$ clause nodes and no literal node was pushed into illusion. Further, notice that since $M$ is a model of $\varphi$, for every variable $p_i$, there is a literal $L_i \in \{p_i, \neg p_i \}$ such that nodes corresponding to $L$ where not used to eliminate illusion from clause nodes. Then, for every variable $p_i$ select such a literal $L_i$. Link all nodes corresponding to $L_i$ to $A_{L_i}$. Finally, observe that for every literal $L$ such that the nodes corresponding to $L$ were not linked to $A_L$, there is at least one node corresponding to $L$ which can be linked to a red node without being pushed into illusion. For every variable $p_i$, link such a node to $E_{p_i}$.   Notice that then the illusion was eliminated from $2|V|$ nodes by adding $4|V|$ edges to the encoding of $\varphi$. So, the illusion has been removed from $|C_{\varphi}|+ 2 |P_{\varphi}|$ nodes without any node being pushed into illusion by adding at most $|C_{\varphi}|+4|P_{\varphi}|$ edges to the encoding. Hence, the defined extension of the encoding of $\varphi$ satisfies the budget and requirement. 
		
	\end{proof}

	\begin{lemma*}\label{lemma:Component1}
		\normalfont \textbf{11.} \textsc{Bounded $q$-Illusion Elimination } is NP-complete for every rational $q \in (0,1)$.
	\end{lemma*}
	
	\begin{proof}
		
		Consider any rational $q \in (0,1)$. First, observe that \textsc{Bounded Removal $q$-Illusion Elimination } is in NP. Let us further construct a network $E^q_{\varphi}$ and a number $k$ for a 2P2N formula $\varphi$ such that the answer to  \textsc{Bounded Removal $q$-Illusion Elimination } for $E^q_{\varphi}$ and $k$ is positive if and only if $\varphi$ is satisfiable. 
		
		In the considered instance we will check whether we can find a subnetwork of $E^q_{\varphi}$ in which edges have been remove between at most  $b_{\varphi}$ pairs of nodes and in which $q$-majority illusion does not hold. The first component of $E^q_{\varphi}$ is $E_{\varphi}$. If $\frac{I_{\varphi} - r_{\varphi}}{|E_{\varphi}|} < q $, then  construct a pump up gadget for $k= h_{|E_{\varphi}|- r_{\varphi}, b_{\varphi}, q}^{\#}$. Otherwise, if  $\frac{I_{\varphi - r_{\varphi}}}{|E_{\varphi}|} \geq q $, construct a pump down gadget for $k= h_{|E_{\varphi}| - r_{\varphi}, b_{\varphi}, q}^{+}$. Let us now show that the answer to  \textsc{Bounded Removal $q$-Illusion Elimination } for $E^q_{\varphi}$ and $b_{\varphi}$ is positive if and only if $\varphi$ is satisfiable.
		
		$(\Rightarrow)$ Suppose that $\varphi$ is satisfiable. We will show that the answer to $E^q_{\varphi}$ and $b_{\varphi}$ is positive.   Let us first consider the case in which $\frac{I_{\varphi} - r_{\varphi}}{|E^q_{\varphi}|} < q $. Then observe that as $\varphi$ is satisfiable, by Lemma 12 we have that it is possible to find a subnetwork $E'_{\varphi}$ of $E_{\varphi}$ in which $b_{\varphi}$ edges are eliminated, and where illusion was eliminated from $r_{\varphi}$ nodes. But then, by Lemma 7 we get that $\frac {I_{\varphi} - r_{\varphi} + h_{|E_{\varphi}|, b_{\varphi}, q}^{+}} {I_{\varphi} - r_{\varphi} + h_{|E_{\varphi}|, b_{\varphi}, q}^{+} +4} <q $. So we can construct a subnetwork of $E^q_{\varphi}$ in which only $b_{\varphi}$ edges are removed but $q$-majority illusion does not hold.
		
		Similarly, if   $\frac{I_{\varphi} - r_{\varphi}}{|E^q_{\varphi}|} \geq q $ we observe that by Lemma 8 we get that  $\frac {I_{\varphi} - r_{\varphi}}  {|E^q_{\varphi} | +h_{|E^q_{\varphi}|, b_{\varphi}, q}^{+}} < q$. So, we get that as we can eliminate illusion from $r_{\varphi}$ nodes in  $E_{\varphi}$ by eliminating $b_{\varphi}$ edges. But then we can construct a subnetwork of $E^q_{\varphi}$ in which only $b_{\varphi}$ edges are removed but $q$-majority illusion does not hold.
		
		$(\Leftarrow)$ Suppose now that $\varphi$ is not satisfiable. We will show that the answer to $E^q_{\varphi}$ and $b_{\varphi}$ is negative. let us first consider the case in which $\frac{I_{\varphi} - r_{\varphi}}{|E^q_{\varphi}|} < q $. Notice that by Lemma \ref{lemma:UpExists1}  that the minimum number of nodes from which illusion needs to be removed for $q$-majority illusion not to hold in $E^q_{\varphi}$ is $r_{\varphi}$. Furthermore observe that in the pump up gadget, the minimum number of edges which is needed to be added to eliminate the illusion from a single node is greater than 4. Moreover, there is a set of nodes $S$ under illusion in $E_{\varphi}$ such that for every $i \in S$, illusion can be eliminated from $i$ by adding 3 edges, without pushing any node under illusion. Further, as $\varphi$ is not satisfiable, by Lemma 12 we get that at it is not possible to remove the illusion from at least $r_{\varphi}$ nodes in $E_{\varphi}$. But then, it is also not possible to remove the illusion from at least $r_{\varphi}$ in $E^q_{\varphi}$.	The reasoning for the case in which $\frac{I_{\varphi} - r_{\varphi}}{|E^q_{\varphi}|} \geq q $ is symmetric.

	\end{proof}

	We will further show that \textsc{Bounded Removal $q$-Illusion Elimination } is NP-complete for every rational $q \in (0,1)$, To show that, we will provide a reduction of \textsc{Bounded Removal $q$-Illusion Elimination } from 2P2N-SAT problem for every such $q$. 
	
	Let us construct the labelled social network which we will call an \emph{encoding} of a formula $\varphi$ in 2P2N form, with a set of variables $P_{\varphi}=\{p_1, \dots , p_m\}$ and a set of clauses $C_{\varphi} = \{C_1, \dots, C_n\}$. We will also refer to it as $E_{\varphi}$. 
	
	\paragraph{Variable, clause, and balance gadgets} Let us construct the encoding of a formula $\varphi$ in 2P2N form, with a set of variables $P_{\varphi}=\{p_1, \dots , p_k\}$ and a set of clauses $C_{\varphi} = \{C_1, \dots, C_l\}$. For convenience, for a literal $L$ we denote as $C^L$ the set of clauses in which $L$ appears.
	Let us start with describing what we call a \emph{variable gadget}. For every variable $p_i \in P_{\varphi}$, let us construct two triples of nodes labelled red, $\{ p_i^1, p_i^2, p_i^3 \}$ and $\{ \neg p_i^1, \neg p_i^2, \neg p_i^3 \}$. We say that the first of them corresponds to the literal $p_i$, while the second to $\neg p_i$, and call members of such triples \emph{literal nodes}. Further, for every literal $L$ construct a node $A_L$, labelled blue, which we call an \emph{auxiliary node} of $L$. For a literal $L$, let $A_L$ be connected to all nodes in the triple corresponding to $L$. Furthermore, for every literal $L$ and every node $L^j$ in the triple corresponding to $L$, construct a group of $|C^L|+2$ red nodes and link them to $L^j$. Moreover, for every such node $n$ linked to $L^j$, create a blue node, link it to $n$ and let all such blue nodes form a clique.  Also, for every variable $p_i$ construct a node $E_i$ labelled blue, which we call an \emph{extra node} of $p_i$. For a variable $p_i$, let $E_i$ be linked to all nodes in the triples corresponding to $p_i$ and $\neg p_i$. Finally, for every variable $p_i$ create five nodes labelled blue, and let them be linked to $E_i$.

	
	\begin{figure}[H]
		\centering
		
		\scalebox{0.6}{	\begin{tikzpicture}
				[->,shorten >=1pt,auto,node distance=1.2cm,
				semithick]

				\node[shape=circle,draw=myred, pattern=crosshatch, pattern color=myred] (A)  {};
				\node[shape=circle,draw=myred, pattern=crosshatch, pattern color=myred, right of=A] (B)  {};
				\node[shape=circle,draw=myred, pattern=crosshatch, pattern color=myred, left of=A] (A')  {};
				\node[draw,  fit=(A') (B)](FIt1) {};
				\node[shape=circle,  left of=A'] (A'')  [fontscale=2] {$p_i$};

				\node[shape=circle,draw=myred, pattern=crosshatch, pattern color=myred, right of=B] (C)  {};
				\node[shape=circle,draw=myred, pattern=crosshatch, pattern color=myred, right of=C] (D)  {};
				\node[shape=circle,draw=myred, pattern=crosshatch, pattern color=myred, right of=D] (D')  {};
				\node[draw,  fit=(C) (D')](FIt2) {};
				\node[shape=circle,  right of=D'] (D'') [fontscale=2] {$\neg p_i$};

				\node[shape=circle,draw=black, below of=B, fill=myblue] (N)  {};
				\node[shape=circle,  right of=N] (N'''') [fontscale=2] {$E_i$};
				
				\draw [thick,-] (N) to  (FIt1) ;
				\draw [thick,-] (N) to  (FIt2) ;
				
				\node[shape=circle,draw=black, below of=N,  fill=myblue] (N')  {};
				\node[shape=circle,  left of=N',  fill=myblue] (k1)  {};
				\node[shape=circle,  left of=k1,  fill=myblue] (k2)  {};
				\node[shape=circle,  right of=N',  fill=myblue] (k3)  {};
				\node[shape=circle,draw=black, right of=k3,  fill=myblue] (N''')  {};
				
				\node[draw,  fit=(k2)  (N''') ](FIt3) {};

				
				\node[shape=circle,draw=black, fill=myblue, above of=A] (E)  {};
				\node[shape=circle,  right of=E] (N''''') [fontscale=2] {$A_{p_i}$};
				\node[shape=circle,draw=black, fill=myblue, above of=D] (F)  {};
				\node[shape=circle,  right of=F] (N'''''') [fontscale=2] {$A_{\neg p_i}$};

				
				\draw [thick,-] (E) to  (FIt1) ;
				\draw [thick,-] (FIt2) -- (F) ;
				\draw [thick,-] (FIt3) to  (N) ;

		\end{tikzpicture}}
		\caption{Variable gadget for the variable $p_i$.}\label{fig:VariableGadget4}
	\end{figure}
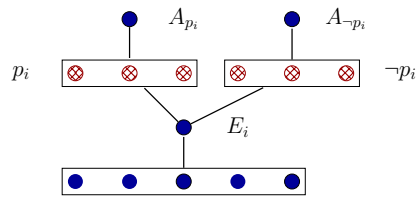

	In addition, let us construct what we call a \emph{clause gadget}. For every clause $C_i \in C_{\varphi}$ let us create a \emph{verifier node} $v_{C_i}$, labelled blue. We say that it corresponds to $C_i$. Furthermore, for every literal $L$ in $C_i$, let all nodes in the triple of literal nodes corresponding to $L$ be linked to the $v_{C_i}$. Finally, for a clause $C_i$, let us create a group of $3 \cdot | P^i  | -1$ nodes labelled blue, connected to $v_{C_i}$, where $P^i$ is the set of literals in $C_i$. It is important to notice that in every network which is a subnetwork of the encoding of $|  \varphi | $ where the illusion is eliminated from  $v_{C_i}$, at least one edge is removed from $v_{C_i}$ and the nodes corresponding to literals in  $C_i$. This condition is also sufficient for the elimination of illusion from $v_{C_i}$. Also, notice that the only node in the clause gadget corresponding to $C_i$ which is under illusion is $v_{C_i}$.
	
	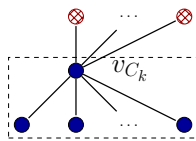
\begin{figure}[H]
		\centering
		
		\scalebox{0.6}{\begin{tikzpicture}
				[->,shorten >=1pt,auto,node distance=1.2cm,
				semithick]
				\node[shape=circle,draw=black, fill=myblue] (A)  {};
				\node[shape=circle,  right of=A] (Q) [fontscale=3]  {$v_{C_k}$};
				\node[shape=circle,draw=black, fill=myblue, below of= A, left of=A] (B)  {};
				\node[shape=circle,draw=black, fill=myblue, below of=A] (C)  {};
				\node[shape=circle,  right of=C] (D)  {$\dots$};
				\node[shape=circle,draw=black, fill=myblue, right of=D] (E)  {};
				
				\draw [thick,-] (A) to  (B) ;
				\draw [thick,-] (A) to  (C) ;
				\draw [thick,-] (A) to  (D) ;
				\draw [thick,-] (A) to  (E) ;
				\node[draw,  dashed, fit=(A) (C)  (B) (C)(E) ](FIt1) {};
				
				\node[shape=circle, draw=myred, pattern=crosshatch, pattern color=myred, above of= A] (A')  {};
				\node[shape=circle,  right of=A'] (B')  {$\dots$};
				\node[shape=circle,draw=myred, pattern=crosshatch, pattern color=myred, right of= B'] (C')  {};
				
				\draw [thick,-] (A) to  (A') ;
				\draw [thick,-] (A) to  (B') ;
				\draw [thick,-] (A) to  (C') ;

		\end{tikzpicture}}
		\caption{Clause gadget for the clause $C_k$.}\label{fig:clausegadget4}
	\end{figure}
	
	Observe that by construction blue is the strict majority colour in the encoding of $\varphi$. Also, let there be no edges in the encoding of $\varphi$ other than those defined before. Observe further that in the encoding of $\varphi$ there exactly $  3|P_{\varphi}| + |C_{\varphi}| $ nodes under illusion.
	
	Also, observe that in each variable gadget, the only nodes under illusion are auxiliary nodes and extra nodes. It is also worth noting that in every subnetwork of the encoding of $\varphi$, in which a node $A_L$ is not under illusion, all edges between $A_L$ and literal nodes are removed. Similarly, in a subnetwork of the encoding of $\varphi$ in which an extra node $E_i$ is not under illusion, at least one edge between $E_i$ and a literal nodes is removed.

	\paragraph{Budget and requirement}Furthermore, we will show that for some specific $k,m \in \mathbb{N}$ it holds that in some subnetwork of $E_{\varphi}$ in which at most $k$ edges are eliminated there are at most $m$ nodes under illusion if and only if $\varphi$ is satisfiable. 
	In the constructed instance of \textsc{Bounded Removal Illusion Elimination } we check the existence of a subnetwork of the encoding of $\varphi$ in which at most $|P_{\varphi }|$ nodes are under illusion, and in which edges are removed not more than $3 \cdot | P_{\varphi}  |$ nodes. We call $|P_{\varphi }|$ \emph{requirement}, and we name  $3 \cdot | P_{\varphi}  |$ as \emph{budget}, or $b_{\varphi}$.  We also say that such a network satisfies the requirement and the budget.

	\paragraph{Observations on modifications satisfying the budget and the requirement} Let us firstly observe that in the encoding of $\varphi$ there exactly $|C_{\varphi}| + 3 \cdot | P_{\varphi} |$ nodes under illusion, namely all clause, auxiliary and extra nodes. Furthermore, observe that the only red nodes they are linked to are literal nodes. Finally, let us notice that in a network satisfying the budget and the requirement the maximum number of nodes from which the illusion is eliminated amounts to the number of clause nodes, extra nodes and half of auxiliary nodes. This holds as, following previous observations, it is sufficient and necessary for a clause node or for an extra node to eliminate one edge between it and a literal node, in order to eliminate the illusion from this node. Furthermore, it is sufficient and necessary to eliminate to eliminate three edges between an auxiliary node and literal nodes to ensure that it is not under illusion. This entails that in a network satisfying the requirement and the budget, illusion is eliminated from all clause nodes, extra nodes and a half of the auxiliary nodes. Also, one can verify that in such a network no node is pushed into illusion. 
	
	In addition, let us observe that in any network satisfying the requirement and the budget, for every variable $p_i$, illusion is eliminated from exactly one of $A_{p_i}$ and $A_{\neg p_i}$. To see that observe that for each literal node at most one edge can be removed from it and any of its neighbours labelled blue, without pushing the literal node into illusion. Therefore, it is not possible to remove the illusion from $E_{p_i}$, if it is removed from both $A_{p_i}$ and $A_{\neg p_i}$. Given a network \textit{SN} satisfying the requirement and the budget, we will say that $p_i$ is false in the network if illusion is eliminated from $A_{p_i}$, and that it is true otherwise. We denote as the \emph{valuation} in \textit{SN} the set of literals true in \textit{SN}. Finally, given a network satisfying the budget and the requirement, a verifier node $v_{C_i}$ and a literal node corresponding to a literal $L$ such that an edge was removed between $v_{C_i}$ and this node, we say that $C_i$ is \emph{satisfied} in the network.
	
	\begin{lemma*}\label{lemma:RemovalWorks}
		\normalfont \textbf{12.} For every formula $\varphi$ in 2P2N form, there is a subnetwork of $E_{\varphi}$ in which at most $k$ edges are eliminated there are at most $m$ nodes under illusion if and only if $\varphi$ is satisfiable. 
	\end{lemma*}
	
	\begin{proof}
		
		Take a formula $\varphi$ in 2P2N form. Let us show that the answer to \textsc{Bounded Removal Illusion Elimination} given the encoding of $\varphi$, as well as budget and requirement is positive if and only if $\varphi$ is satisfiable.
		
		$(\Rightarrow)$ Suppose that $\varphi$ is not satisfiable. Let us further suppose towards contradiction that there exists a subnetwork \textit{SN} of the encoding of $\varphi$ satisfying the budget and the requirement. As $\varphi$ is not satisfiable, for every valuation $V$ over $P_{\varphi}$ there is a clause $C_i$ such that all literals in $C_i$ are false in $V$. Further, consider the valuation in \textit{SN}.  Then notice that, by previous observations, illusion is eliminated in \textit{SN} from all verifier nodes by removing edges between them and literal nodes. Therefore, all clauses are satisfied in \textit{SN}. But then, there is a valuation over $P_{\varphi}$ under which $\varphi$ is true, which contradicts the assumptions.
		
		$(\Leftarrow)$ Suppose that $\varphi$ is satisfiable. Let us construct a subnetwork \textit{SN} of the encoding of $\varphi$ satisfying the requirement and the budget. As $\varphi$ is satisfiable, there exists a valuation $V$ over $P_{\varphi}$ such that for every clause $C_i$ there is a literal $L$ in $C_i$ which is true in $V$. Then, let us eliminate edges between all literal nodes corresponding to literals false in $V$, and auxiliary nodes. Then, for every literal $L$ true in \textit{SN}, eliminate an edge from the node $L^1$  and an extra node. Finally, for every clause $C_i$, eliminate an edge between $v_{C_i}$ and exactly one literal node corresponding to a literal true in \textit{SN}, represented in $C_i$, which is always possible since $V$ is a model of $\varphi$. Notice that $\varphi$ is in 2P2N form. Therefore, as $L$ occurs twice in $\varphi$, we can ensure that at most one edge is eliminated between a node corresponding to $L$ and a verifier node. But then, in the constructed subnetwork, only $|  P_{\varphi}| $ nodes are under illusion, and edges have been eliminated between $3 \cdot |P_{\varphi}|$ pairs of nodes. Thus,  the encoding of $\varphi$ satisfies the requirement and the budget.
	\end{proof}

	We are now ready to show that \textsc{Bounded Removal $q$-Illusion Elimination } is NP-complete for every rational $q \in (0,1)$. Let $I_{\varphi}$ denote the number of nodes under illusion in $E_{\varphi}$.
	
	\begin{lemma*}\label{lemma:Component2}
		\normalfont \textbf{13.} \textsc{Bounded Removal $q$-Illusion Elimination } is NP-complete for every rational $q \in (0,1)$.
	\end{lemma*}
	
	\begin{proof}
		Consider any rational $q \in (0,1)$. First, observe that \textsc{Bounded Removal $q$-Illusion Elimination } is in NP. Let us further construct a network $E^q_{\varphi}$ and a number $k$ for a 2P2N formula $\varphi$ such that the answer to  \textsc{Bounded Removal $q$-Illusion Elimination } for $E^q_{\varphi}$ and $k$ is positive if and only if $\varphi$ is satisfiable. 
		
		In the considered instance we will check whether we can find a subnetwork of $E^q_{\varphi}$ in which edges have been remove between at most  $b_{\varphi}$ pairs of nodes and in which $q$-majority illusion does not hold. The first component of $E^q_{\varphi}$ is $E_{\varphi}$. If $\frac{I_{\varphi} - r_{\varphi}}{|E_{\varphi}|} < q $, then  construct a pump up gadget for $k= h_{|E_{\varphi}|- r_{\varphi}, b_{\varphi}, q}^{\#}$. Otherwise, if  $\frac{I_{\varphi - r_{\varphi}}}{|E_{\varphi}|} \geq q $, construct a pump down gadget for $k= h_{|E_{\varphi}|, b_{\varphi}, q}^{+}$. Let us now show that the answer to  \textsc{Bounded Removal $q$-Illusion Elimination } for $E^q_{\varphi}$ and $b_{\varphi}$ is positive if and only if $\varphi$ is satisfiable.
		
		$(\Rightarrow)$ Suppose that $\varphi$ is satisfiable. We will show that the answer to $E^q_{\varphi}$ and $b_{\varphi}$ is positive.   Let us first consider the case in which $\frac{I_{\varphi} - r_{\varphi}}{|E^q_{\varphi}|} < q $. Then observe that as $\varphi$ is satisfiable, by Lemma 12 we have that it is possible to find a subnetwork $E'_{\varphi}$ of $E_{\varphi}$ in which $b_{\varphi}$ edges are eliminated, and where illusion was eliminated from $r_{\varphi}$ nodes. But then, by Lemma 8 we get that $\frac {I_{\varphi} - r_{\varphi} + h_{|E_{\varphi}|, b_{\varphi}, q}^{+}} {I_{\varphi} - r_{\varphi} + h_{|E_{\varphi}|, b_{\varphi}, q}^{+} +4} <q $. So we can construct a subnetwork of $E^q_{\varphi}$ in which only $b_{\varphi}$ edges are removed but $q$-majority illusion does not hold.
		
		Similarly, if   $\frac{I_{\varphi} - r_{\varphi}}{|E^q_{\varphi}|} \geq q $ we observe that by Lemma 7 we get that  $\frac {I_{\varphi} - r_{\varphi}}  {|E^q_{\varphi} | +h_{|E^q_{\varphi}|, b_{\varphi}, q}^{+}} < q$. So, we get that as we can eliminate illusion from $r_{\varphi}$ nodes in  $E_{\varphi}$ by eliminating $b_{\varphi}$ edges. But then we can construct a subnetwork of $E^q_{\varphi}$ in which only $b_{\varphi}$ edges are removed but $q$-majority illusion does not hold.
		
		$(\Leftarrow)$ Suppose now that $\varphi$ is not satisfiable. We will show that the answer to $E^q_{\varphi}$ and $b_{\varphi}$ is negative. let us first consider the case in which $\frac{I_{\varphi} - r_{\varphi}}{|E^q_{\varphi}|} < q $. Notice that by Lemma 7 we get that we get that the minimum number of nodes from which illusion needs to be removed for $q$-majority illusion not to hold in $E^q_{\varphi}$ is $r_{\varphi}$. Furthermore observe that in the pump up gadget, the minimum number of edges which is needed to be removed to eliminate the illusion from a single node is 4. Moreover, there is a set of nodes $S$ under illusion in $E_{\varphi}$ such that for every $i \in S$, illusion can be eliminated from $i$ by removing 3 edges, without pushing any node under illusion. Further, as $\varphi$ is not satisfiable, by Lemma \ref{lemma:RemovalWorks} we get that at it is not possible to remove the illusion from at least $r_{\varphi}$ nodes in $E_{\varphi}$. But then, it is also not possible to remove the illusion from at least $r_{\varphi}$ in $E^q_{\varphi}$.	
		
	\end{proof}
	
	\begin{theorem}\label{theorem:ExtraRemoval}
		\textsc{Addition $q$-Illusion Elimination } and \textsc{removal $q$-Illusion Elimination } are NP-complete for every rational $q \in (0,1)$.
	\end{theorem}
	
	\begin{proof}
		The claim follows immediately from Lemma \ref{lemma:Component1} and Lemma \ref{lemma:Component2}.
	\end{proof}
	
	\subsection*{Plurality illusion}
	
	Let us show that there are networks which allow for colouring with multiple colours where all agents perceive an option different than the plurality winner as the most popular option, but which do not admit 1-majority illusion. Let us first provide a definition of a \emph{plurality illusion}. Let $C$ be a set of colours. Given a labelled social network \textit{SN}=$(N,E,f)$, where $f$ is a function $f: N \rightarrow C$,  we denote the set of most popular colours in the network as $\textit{Pl}_{\textit{SN}}$. If the most popular colour is unique, we will call it the \emph{plurality winner}. Formally, $Pl_{\textit{SN}}= \displaystyle \argmax_{c \in C} |\{i \in N : f(i)=c\}|$. Similarly, for an agent $i$, $\textit{Pl}_{\textit{SN}}^i$ is the set of most popular options in $i$'s neighbourhood. Formally, $\textit{Pl}_{\textit{SN}}^i= \displaystyle \argmax_{c \in C} |\{i \in {\it E }(i) \mid f(i)=c\}|$. If $\textit{Pl}_{\textit{SN}}^i=c$ for some $c \in C$, we say that $c$ is the plurality winner in $i$'s neighbourhood.  Then, we say that an agent $i \in N$ is under illusion plurality illusion if plurality winner in $i$'s neighbourhood is different than the plurality winner (while both exist). Further, $f$ induces plurality illusion if all agents are under plurality illusion in $(N,E,f)$. Also, we say that $(N,E)$ \emph{admits} plurality illusion if some $f:N \rightarrow C$ induces plurality illusion.
	
	\begin{observation}
		There are networks which admit a plurality illusion, but not majority illusion.
	\end{observation}

	\begin{figure}[H]	
		\centering
		\scalebox{0.6}{\begin{tikzpicture}
				[->,shorten >=1pt,auto,node distance=1.2cm,
				semithick]
				\node[shape=circle,draw=black, pattern=crosshatch, pattern color=myred] (A)  {};
				\node[shape=circle,draw=black,  right of=A,pattern=crosshatch, pattern color=myred] (B)  {};
				\node[shape=circle,draw=black, below of =A, fill=green] (C)  {};
				\node[shape=circle,draw=black, right of =C, fill=green] (D)  {};
				
				\node[shape=circle,draw=black, above of =A, fill=blue] (E)  {};
				\node[shape=circle,draw=black, above of =B, fill=blue] (F)  {};					
				\node[shape=circle,draw=black, below of =C, fill=blue] (G)  {};
				\node[shape=circle,draw=black, below of =D, fill=blue] (H)  {};
				
				\draw [thick,-] (A) to  (E) ;
				\draw [thick,-] (B) to  (F) ;
				\draw [thick,-] (C) to  (G) ;
				\draw [thick,-] (D) to  (H) ;
				
				\draw [thick,-] (A) to  (B) ;
				\draw [thick,-] (A) to  (C) ;
				\draw [thick,-] (A) to  (D) ;				
				\draw [thick,-] (B) to  (C) ;
				\draw [thick,-] (B) to  (D) ;
				\draw [thick,-] (C) to  (D) ;
				
				\node[shape=circle,draw=black, right of = B ,fill=green] (I)  {};
				\node[shape=circle,draw=black, below of = I ,fill=blue] (J)  {};
				\node[shape=circle,draw=black, above of = I ,pattern=crosshatch, pattern color=myred] (K)  {};
				\node[shape=circle,draw=black, right of = K ,pattern=crosshatch, pattern color=myred] (L)  {};
				\node[shape=circle,draw=black, right of = L ,fill=green] (M)  {};
				
				\draw [thick,-] (I) to  (J) ;
				\draw [thick,-] (I) to  (K) ;
				\draw [thick,-] (I) to  (L) ;
				\draw [thick,-] (I) to  (M) ;
				
		\end{tikzpicture}}

		\caption{Example of a social network admitting a 1-plurality illusion with three colours, but not admitting a 1-majority illusion.}\label{fig:PlNotMaj}
	\end{figure}
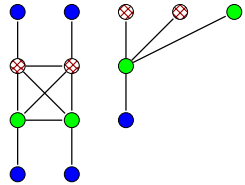
	
	\begin{example}
		
		Consider the undirected graph shown in Figure \ref{fig:PlNotMaj}. Let us begin with showing that this network admits a 1-plurality illusion with three colours. To see that, consider the labelling in Figure \ref{fig:PlNotMaj}. Notice that 5 nodes are labelled blue, 4 labelled red and 4 labelled green. Thus, blue is the plurality winner. However, one can verify that there is a plurality winner other than blue in the neighbourhood of every node in the network. So, the proposed labelling is a plurality illusion.
		
		Furthermore, let us demonstrate that this network does not admit 1-majority illusion. Suppose towards contradiction that there is a labelling $L$ of this network which is a majority illusion. Then observe that there are 13 nodes in the network, and hence at most 6 nodes are labelled red. Moreover, all nodes in the 4-clique in the left subnetwork and the central node in the right subnetwork have dependants, and hence need to be labelled red. Furthermore, the central node in the right subnetwork has four neighbours, and hence by assumption it is under the majority illusion, at least three of the nodes linked to it need to be labelled red. But then at least 8 nodes are labelled red in $L$, which contradicts the assumptions. So, the network in Figure \ref{fig:PlNotMaj} does not admit a 1-majority illusion.
		
	\end{example}

\end{document}